\documentclass[11pt]{amsart}

\usepackage{macros}

\spacing{1.1}

\author{Natalie M. Paquette}
\address{Institute for Advanced Study, School of Natural Sciences, Princeton NJ, 08540, USA}
\address{Department of Physics, University of Washington, Seattle, WA, 98195-1560, USA}
\email{npaquett@uw.edu}
\author{Brian R. Williams}
\address{School of Mathematics, University of Edinburgh, Edinburgh, UK}
\email{brian.williams@ed.ac.uk}

\date{\today}
\title{Koszul duality in quantum field theory}

\def\id{\mathbb{1}}

\begin{document}
\maketitle

\begin{abstract}
In this article, we introduce basic aspects of the algebraic notion of Koszul duality for a physics audience. We then review its appearance in the physical problem of coupling QFTs to topological line defects, and illustrate the concept with some examples drawn from twists of various simple supersymmetric theories. Though much of the content of this article is well-known to experts, the presentation and examples have not, to our knowledge, appeared in the literature before. Our aim is to provide an elementary introduction for those interested in the appearance of Koszul duality in supersymmetric gauge theories with line defects and, ultimately, its generalizations to higher-dimensional defects and twisted holography. 
\end{abstract}

\section{Introduction}

The connections between string theory and mathematics are as strong as ever.  As history has shown time and again, a deep understanding of theoretical physics often requires (and sometimes begets!) beautiful mathematics. In recent years, the subject of higher algebras has been drawn into this web of relations between the two disciplines. A ubiquitous duality in mathematics called Koszul duality is part of the homological algebraic vanguard emerging in theoretical physics in recent years. We will provide a brief, example-oriented introduction to Koszul duality in \S \ref{s:review}. 

The appearance of Koszul duality parallels a recent surge of interest in the study of extended operators, defects, and boundary conditions in quantum field theories and, to be sure, this is no coincidence. 
As we describe in this note, Koszul duality captures, from an operator algebraic perspective, essential symmetries of coupled bulk/defect systems compatible with gauge invariance. 

Let us give a heuristic for the context in which Koszul duality appears in a broader field theoretic context. 
Along the way, we will preview the remainder of this article. 
Suppose that $\cA$ is the algebra of local operators corresponding to some QFT. 
To obtain this algebra structure we assume that our QFT is defined on a manifold of the form $\RR \times M$. 
The algebra structure arises from the OPE in the $\RR$-direction.  

Next, assume that $\cB$ is another algebra, which we think of as being associated to some quantum mechanical system along $\RR \times \{x\} \subset \RR \times M$ where $x \in M$. 
For now, and for the most of this work, we assume that $\cA,\cB$ are simply associative (possibly up to homotopy) algebras. 
We will always keep track of an internal grading by ghost degree and we denote by $\delta$ the BRST operator. 

The natural problem we contemplate is what it means to couple the two systems $\cA$ and $\cB$. 
The tensor product algebra $\cA \otimes \cB$ should be thought of as the algebra of operators of the total, uncoupled system, along $\RR$. 
By a {\em coupling} of the two systems we will mean a {\em Maurer--Cartan element} in this algebra
\[
\alpha \in \cA \otimes \cB .
\]
That is, $\alpha$ is an element of ghost degree one which satisfies the Maurer--Cartan equation
\[
\delta \alpha + \alpha \star \alpha = 0 .
\]
In \S \ref{s:algcouple}, we will elaborate further on this definition. 
In \S \ref{s:lagcoupling}, we will explain the relationship of this algebraic notion of a coupling to the more standard description in terms of local Lagrangians.

We arrive at Koszul duality from this description of a coupling between two systems. 
The key observation is that the data of the Mauer--Cartan element $\alpha$ is equivalent to the data of a map of algebras
\[
\alpha \colon \cA^! \to \cB 
\]
where $\cA^!$ is Koszul dual to the algebra $\cA$. 
From this observation, the Koszul dual has the following interpretation: $\cA^!$ is the algebra of operators on the {\em universal} line defect supported on $\RR \times \{x\}$. We will explain this point of view in detail in \S \ref{s:lines} and illustrate it concretely in a free field theory example in \S \ref{s:free}. 

We stress that Koszul dual algebras at the classical level may be (dually) deformed when considering quantum effects. The Maurer--Cartan equation is linear at the classical level and simply captures BRST-closure of the coupling: $\delta \alpha = 0$. However, a classical coupling may fail to be BRST invariant at the quantum level. As we will explain, this manifests as a deformation of the BRST operator $\delta \rightarrow \delta + \alpha \star (-)$, which in turn deforms the Maurer--Cartan equation by the quadratic term as above. Requiring that the coupling satisfies the full, nonlinear Maurer--Cartan equation will produce a noncommutative deformation of the classical operator product. We will illustrate this in an example in \S \ref{s:crit}. For another illustrative example of this phenomenon, in which the algebra $\mathfrak{g}[[z]]$ deforms to the Yangian due to the presence of a two-loop gauge anomaly, see \cite{CWY}. In general, quantum corrections may take the form of an infinite series, with nontrivial corrections at every loop order. 

Finally, in \S \ref{s:conclusions}, we will take a somewhat broader view. We will explain some connections between Koszul duality and other areas of inquiry, mention other appearances and applications of Koszul duality in the field theory literature, and highlight further examples of our perspective on Koszul duality, as well as some generalizations, in both field theory and holography. 

\subsection*{Notations \& Conventions}

We briefly summarize some conventions we use throughout this note. Unless otherwise stated, we study quantum field theories in Euclidean signature and on flat spacetimes. 

In particular, all of our considerations are local, and are carried out at the perturbative level. We do not account for interesting global properties of quantum field theories, such as studying the theories on manifolds of nontrivial topology, nor do we account for potentially interesting nonperturbative effects, except for brief comments in \S \ref{s:conclusions}. 
Clarifying any interactions between Koszul duality and these effects is a topic we leave for future work. 

By a grading we will mean a $\ZZ$-grading unless otherwise specified. 
Throughout this note we use the {\em cohomological} grading convention.
A cochain complex is a graded vector space $A$ equipped with a sequence of operators $\d_i \colon A^i \to A^{i+1}$ which satisfies $\d_{i+1} \circ \d_i = 0$. 
Given a graded vector space $A$ and an integer $n$, denote by $A[n]$ the new graded vector space which satisfies $A[n]^i := A^{n+i}$. 
In other words, $A[n]$ is like $A$ where the grading is shifted down by the integer $n$.

We will use the Koszul rule of signs for all graded algebraic structures.
For example, a graded commutative algebra is a graded vector space equipped with a degree zero product $\cdot$ which satisfies $a \cdot b = (-1)^{|a||b|} b \cdot a$. 
The bracket in a graded Lie algebra follows the Koszul sign rule $[a, b] = (-1)^{|a||b|+1} [b, a]$.

\section*{acknowledgements} 

There is a rich history of how ideas of Koszul duality in QFT developed and we are undoubtedly unaware of many important contributions. 
We learned much of this in conversations with Kevin Costello and with fellow collaborators.
We are especially grateful to Kevin for teaching us many of the ideas and perspectives described in this note, for enjoyable and inspiring collaborations, and for emphasizing to us and others the importance of Koszul duality in mathematical physics. 
We also wish to thank Tudor Dimofte, Chris Elliott, Davide Gaiotto, Owen Gwilliam, Justin Hilburn, Si Li, Ingmar Saberi, and Philsang Yoo for numerous useful conversations and collaborations throughout the years. 
Finally, we are grateful to K.~Costello, R.~Grady, O.~Gwilliam, A.~Khan, and M.~Szczesny for very helpful comments on a draft of this note. 

NMP was supported by the grant NSF PHY-1911298, and the Sivian Fund, and is currently supported by the DOE's Research Opportunities in High Energy Physics and the University of Washington. 
BRW thanks the University of Edinburgh for their support.

\section{A short review of Koszul duality} \label{s:review}

We provide a basic introduction to Koszul duality for associative algebras, emphasizing the most simple algebraic examples. Since our aim is expository, we will omit many important references and generalizations. 
For a more exhaustive textbook introduction to Koszul duality in mathematics see, for example, \cite[Chapter 3]{LV}. 

\subsection{Some definitions}

Let $\cA$ be an associative algebra.
We will often work in the more general situation where $\cA$ is a differentially graded (dg) algebra, but let us focus on this simple case first. 
Physically, $\cA$ represents the algebra of operators of a quantum mechanical system. 


An $\cA$-module $M$ also has a quantum mechanical interpretation. 
Physically, if $\cA$ represents the algebra of local operators of a quantum mechanical system, then an $\cA$-module $M$ labels a point defect at a particular point in time. 
If the point defect is contained in the interior of the one-dimensional manifold then one needs a bi-module structure on $M$. 
 
A map of algebras $f \colon \cA \to \cB$ is a map of cochain complexes which also preserves the associative products.
Notice that $f$ endows $\cB$ with the structure of an $\cA$-module by the formula $a \cdot b \define f(a) \cdot b$, for $a \in \cA$ and $b \in \cB$, where on the right hand side we use the associative product on $\cB$. 

Similarly if $M,N$ are $\cA$-modules one can define a map of modules $f \colon M \to N$ to be a linear map which intertwines the $\cA$-module structure. 
All $\cA$-modules combine to form a linear category, meaning the morphism sets are endowed with the structure of a vector space over $\CC$.
We will be interested in an enhancement of this to a {\em dg category}, which means its set of morphisms is equipped with the structure of a cochain complex (so, the space of Homs between any two fixed modules is graded and equipped with a differential). 

For $\cA$-modules $M,N$ we denote by $\RHom_\cA(M,N)$ the dg vector space of homomorphisms. 
We will not use the full technology of derived categories in this note, but we remark on some important features to give the reader some intuition. 
The $i$th cohomology of the complex $\RHom_{\cA} (M,N)$ comprises the Ext groups ${\rm Ext}^i_{\cA}(M,N)$. 
For all $M,N$ the group ${\rm Ext}^0_{\cA}(M,N)$ agrees with the naive space of $\cA$-linear homomorphisms $M \to N$. 
The space ${\rm Ext}^1_{\cA}(M,N)$ is also relatively easy to understand: it is the space of equivalence classes\footnote{Two extensions $E,E'$ are equivalent if there is an isomorphism of $\cA$-modules $E \to E'$ which can be extended to an isomorphism of short exact sequences.} of short exact sequences of $\cA$-modules of the form
\[
0 \to N \to E \to M \to 0 .
\]

When $M = N$ the dg vector space (or cochain complex) $\RHom_\cA(M,M)$ acquires the structure of a differentially graded (dg) algebra.
A dg algebra is $\ZZ$-graded algebra equipped with a square-zero, degree one operator $\d \colon \cA \to \cA[1]$ such that 
\[
\d (a \cdot b) = (\d a) \cdot b + (-1)^{|a|} a \cdot (\d b) 
\]
where $|a|$ denotes the degree of the element $a$. 
The zeroth cohomology of $\RHom_\cA(M,M)$ agrees with the algebra of $\cA$-linear endomorphisms ${\rm End}_{\cA}(M)$.  
In general, to compute the space of derived homomorphisms one needs to first \textit{freely resolve} the module $M$, and we will see an example shortly.

A particularly important structure on an algebra will be that of an augmentation. 
An {\em augmentation} on an algebra $\cA$ is a map of algebras 
\[
\ep \colon \cA \to \CC 
\]
meaning that $\ep (a b) = \ep(a) \ep(b)$ and $\ep (\d a) = 0$. This augmentation may have a nontrivial kernel, which acts trivially on $\CC$. 
In other words, an augmentation defines the data of an $\cA$-module structure on the one-dimensional module. 
A choice of augmentation corresponds physically to a choice of vacuum, see \S \ref{sec:univ}. 
In most examples we encounter in this note, the augmentation is essentially unique; we plan to expand on explicit examples with vacuum degeneracy in subsequent work.
When we want to stress this interpretation of an augmentation, we denote the one-dimensional module by~$\CC_\ep$.

The Koszul dual of the algebra $\cA$ is the dg algebra
\[
\cA^! \define \RHom_\cA(\CC_\ep,\CC_\ep) .
\] 
Notice that the Koszul dual of $\cA$ is defined with respect to an augmentation $\ep$, though when the augmentation is understood we will omit it from the notation. Roughly speaking, this denotes that the Koszul dual algebra is the space of symmetries of the trivial module that commute with $\cA$. Further, notice that even if a dg-algebra $\cA$ is an ordinary associative algebra with trivial differential, its Koszul dual will be a complex with support in cohomological degrees $\geq 1$.

This construction of the Koszul dual makes sense in the greater generality where $\cA$ itself is a dg algebra. 
In physical contexts, where we think about $\cA$ as being the algebra of operators of a quantum mechanical system, the operator $\d$ should be thought of as a BRST operator or a square-zero supercharge. 
The grading will be by ghost number or fermion number, respectively \footnote{In many physical and mathematical contexts, we only require a $\ZZ/2\ZZ$ grading, e.g. fermion parity, but we will restrict to the $\ZZ$-graded case for now.}.
If we forget the product then $(\cA, \d)$ has the structure of a cochain complex. 
We work over the complex numbers $\CC$ so all products and maps are required to be $\CC$-linear. 

Concretely, there is a model for the Koszul dual $\cA^!$ called the {\em bar complex}.
See \cite{Weibel}, for a textbook account. 
We won't use it explicitly here, but we will compute some basic examples in a low-brow fashion momentarily. 

There is a canonical map of dg algebras $\cA \to (\cA^!)^!$. 
In nice situations\footnote{For instance, if $H^0 (\cA) \cong \CC$, $H^{<0} (\cA) = 0$, and $H^i(\cA)$ is finite-dimensional for all $i > 0$.} 
this map is a quasi-isomorphism, meaning it is an isomorphism on cohomology. 
Thus, with the choice of augmentation, $\CC_\ep$ is simultaneously a module for $\cA, \cA^!$, and these algebras have commuting actions on $\CC_\ep$.

Throughout this note, we will focus on the computation of Koszul dual algebras. 
However, we note that modules for these algebras are also of physical and mathematical interest. 
One of the key characterizing features of Koszul duality is that certain subcategories of the derived categories of dg-modules of $\cA$ and $\cA^!$ are equivalent \cite{BGG, BGSch, BGSoerg, Pos}. \footnote{Certain assumptions on the dg algebra $\cA$ are required for this result to hold. 
To get an idea of the technical homological algebra that goes into this, even for examples of physical interest, we refer to \cite[Section 8]{CYangian}. } 
We will discuss modules briefly in \S \ref{s:conclusions}, but it will not be our main focus. 

\subsection{Free field example} 

To illustrate the abstract considerations above, we present a short computation of the most famous instance of Koszul duality. 
This is the case where the algebra is the free symmetric algebra on a single (bosonic) generator, $\cA = \CC[x]$. 
The augmentation $\ep$ is the obvious one which sends a polynomial to its constant coefficient. 

The resulting one-dimensional module $\CC_{\ep}$ is not free (i.e. it does not admit a presentation in terms of linearly independent basis elements) as a $\CC[x]$-module, so, in order to make efficient computations, we need to find a cochain complex of free modules which resolves it. 
The prototypical example of a free module is $\CC[x]$ itself.

Using the technique of resolving a module as cochain complex is not unheard-of in physics: indeed, BRST quantization for some gauge group $G$ involves making a so-called projective resolution of the space of $G$-invariant local operators. 
The corresponding complex introduces fields of nonzero ghost number and a BRST differential. 
The cohomology of the BRST complex reproduces $G$-invariant local operators concentrated in degree 0 (the physical observables), and cohomologies in higher ghost number vanish.\footnote{In mathematical language, the BRST complex is a model for derived invariants, which may differ from strict invariants, but in physics one is often interested in the strict case.}

Similarly, the Batalin--Vilkovisky (BV) formalism provides a derived enhancement of the critical locus, i.e. solutions of the Euler-Lagrange equations, of a Lagrangian quantum field theory using antifields and antighosts. 
This is a geometric avatar of the Koszul--Tate complex. 

In general, a free resolution of some $R$-module $M$ is an exact sequence of $R$-modules of the form
\[
\cdots \to M_2 \to M_1 \to M_0 \to M \to 0 .
\]
Said differently, $(\cdots M_2 \to M_1 \to M_0)$ has the structure of a cochain complex with $M_i$ in degree $-i$. 
The exact sequence endows a cochain map of $R$-modules from this cochain complex to the original module $M$ which is an isomorphism in cohomology.
Such a cochain map is called a quasi-isomorphism. 
Quasi-isomorphism plays nicely with Koszul duality in the following sense. 
If $\cA_1, \cA_2$ are quasi-isomorphic algebras, via a quasi-isomorphism which preserves augmentations, then $\cA_1^{!}, \cA_2^{!}$ are quasi-isomorphic as well. 

Different quasi-isomorphic representations of the same module can have useful physical interpretations, even though the physical observables of interest involve the passage back to cohomology. 
To take a particularly simple instance, Yang-Mills theory can be quantized using the minimal BRST complex or using the full machinery of the BV-BRST formalism, but we do not need to use the latter for most purposes.

A natural choice for a free resolution in this example is the so-called ``Koszul complex'', a progenitor of the Koszul-Tate complex in BV-quantization, defined as the polynomial algebra on an bosonic generator $x$ and a fermionic generator $\xi$. 
There is a nontrivial differential defined by $\d \xi = x$. 
A presentation for this dg algebra is 
\beqn\label{eqn:res1}
\bigg(\CC[x, \xi] \, , \, \d = x \frac{\partial}{\partial \xi} \bigg).
\eeqn
On linear generators, the cochain complex is of the form
\[
\begin{tikzcd} 
1 & x & x^2 & x^3 & \cdots \\
& \xi \ar[u] & x \xi \ar[u] & x^2 \xi \ar[u] & \cdots 
\end{tikzcd}
\]
where the bottom line is in cohomological degree $-1$, the top line is in cohomological degree zero, and the arrows denote acting by the differential as usual. 

The derived space of homomorphisms of $\CC_\ep$ is equal to linear endomorphisms 
\[
\CC[x, \xi] \to \CC[x,\xi]
\] 
which commute with the $\CC[x]$-module structure: that is, endomorphisms which commute with multiplication by $x$\footnote{The $\CC[x]$-module structure is by multiplication, so this is what it means to be a map of $\CC[x]$-modules.}. 
Such endomorphisms are algebraically generated by multiplication by $x$, multiplication by $\xi$, and differentiation by $\xi$. 
Thus 
\beqn\label{eqn:rhom1}
\RHom_{\CC[x]} \left(\CC_\ep, \CC_\ep\right) = \left(\CC\left[x, \xi, \frac{\partial}{\partial \xi} \right] \, , \, \d = \left[x \frac{\partial}{\partial \xi} , -\right] \right) .
\eeqn
The differential is given by taking the commutator with the original differential of the Koszul complex. 

Now, let $\eta$ denote an odd parameter and consider the cochain complex equipped with the trivial differential
\[
 \left(\CC[\eta] \, , \, \d = 0 \right) .
\]
Since $\eta^2=0$, this is simply the graded vector space linearly spanned by a single element in even degree, $1$, and a single element in odd degree, $\eta$. 

Consider the following map of dg algebras
\[
\left(\CC[\eta] \, , \, \d = 0 \right) \hookrightarrow \RHom_{\CC[x]} \left(\CC_\ep, \CC_\ep\right) 
\]
which sends $1 \mapsto 1$ and $\eta \mapsto \frac{\partial}{\partial \xi}$. 
This is a map of cochain complexes since $\left[x \frac{\partial}{\partial \xi} , \frac{\partial}{\partial \xi}\right] = 0$.

Further, this map induces an isomorphism on cohomology. 
To see this, let us compute the cohomology of our model in \eqref{eqn:rhom1}. 
Notice that the algebra of closed elements, those killed by the differential, is freely generated by $x, \frac{\partial}{\partial \xi}$. 
Furthermore, if $n \geq 0$ and $k = 0,1$ the element $x^{n+1}\left(\frac{\partial}{\partial \xi}\right)^k$ is exact since 
\[
\left[x \frac{\partial}{\partial \xi} , x^n \xi \left(\frac{\partial}{\partial \xi}\right)^k \right] = x^{n+1} \left(\frac{\partial}{\partial \xi}\right)^k . \\
\]
This shows that the cohomology is algebraically generated by $\frac{\partial}{\partial \xi}$. 
The result follows via the identification $\eta = \frac{\partial}{\partial \xi}$. 

In other words, we see that the Koszul dual of the polynomial algebra on a bosonic generator $x$ is equivalent to the polynomial algebra on a fermionic generator $\eta = \frac{\partial}{\partial \xi}$
\[
\CC[x]^! \simeq \CC[\eta] .
\]

In a completely analogous way, one proves that the Koszul dual of the algebra of polynomials in $n$ bosonic variables $\{x_1,\ldots, x_n\}$, which form a basis for a vector space $V$, is the algebra of polynomials in $n$ fermionic variables $\{\eta^1,\ldots,\eta^n\}$. 
In other words, there is an equivalence
\[
{\rm S}^\bu (V)^! \simeq \wedge^\bu V^\vee .
\]
The linear dual on the right-hand side $V^{\vee}$ arose in the computation above; the fermionic generator $\eta^i = \frac{\partial}{\partial \xi_i}$ is naturally dual to the generator $\xi_i$. 

Instead of working with just even or odd variables in the above calculation, we can introduce an integer cohomological grading. 
Declare the original vector space $V$ to sit in cohomological degree zero. 
Then on the dual side, $V^\vee$ must sit in cohomological degree $+1$.
In other words, if we start with a polynomial algebra on variables $x_1,\ldots,x_n$ of degree zero then the dual is the free graded algebra on generators $\eta^1,\ldots,\eta^n$ of degree~$+1$.

\subsection{An example from gauge theory}

Let us compute another simple, famous example of a Koszul dual pair.
Let $\fg$ be a Lie algebra with basis $\{x_1,\ldots, x_n\}$ and structure constants $\{f_{ij}^k\}$. 
The enveloping algebra of $\fg$, denoted $U \fg$, is the algebra generated by the bosonic generators $x_1,\ldots,x_n$ subject to the relation 
\beqn \label{eqn:relation}
x_i x_j - x_j x_i = f_{ij}^k x_k .
\eeqn
There is a natural augmentation on this algebra $\ep \colon U \fg \to \CC$ which sends the linear generators to zero $x_i \mapsto 0$ and the constant $1 \in U \fg$ to $1 \in \CC$. 
This equips us with the $U \fg$-module structure on $\CC_\ep$ where $x_i$ acts by zero and $1$ acts by the identity. 

If we forget the algebra relation, the computation of the Koszul dual of $U \fg$ is identical to the calculation with the polynomial algebra above. 
This follows from the Poincar\'e--Birkoff--Witt theorem: there is an isomorphism of $U\fg$ modules $U \fg \cong \Sym (\fg)$. 
In other words, we can write a presentation for our algebra as
\[
U \fg = \CC[x_i] / \sim 
\]
with the equivalence relation generated by \eqref{eqn:relation}. 

The Koszul complex freely resolving $\CC$ as a $U \fg$-module also takes into account the structure constants of the Lie algebra. 
It is defined by
\[
\left(\CC[x_i,\xi_i] / \sim \, , \,\d =  x_i \frac{\partial}{\partial \xi_i} \right) 
\]
where we in addition to the relation \eqref{eqn:relation} we also impose the relation
\beqn \label{eqn:relation2}
x_i \xi_j - \xi_j x_i = f_{ij}^k \xi_k .
\eeqn
It is easy to see that the differential is compatible with the relations.

The space of derived homomorphisms of the module $\CC_\ep$ is computed in a similar way as in the abelian example. 
It is given by the following complex 
\beqn
\label{eqn:big}
\left(\CC\left[x_i,\xi_i, \frac{\partial}{\partial \xi_i}\right] / \sim \, , \, \d = \left[x_i \frac{\partial}{\partial \xi_i} , - \right] \right) 
\eeqn
where, in addition to \eqref{eqn:relation} and \eqref{eqn:relation2} we have introduced the third relation
\beqn \label{eqn:relation3}
x_i \frac{\partial}{\partial \xi_j} - \frac{\partial}{\partial \xi_j} x_i = f_{ij}^k \frac{\partial}{\partial \xi_i} .
\eeqn

This is a big cochain complex, but just as in the abelian case we can find a smaller model for it.
Consider the free graded algebra generated by $n$ fermionic variables $\eta^1, \ldots, \eta^n$ equipped with the following differential
\beqn\label{eqn:ce1}
\left(\CC[\eta^i] \, , \, \d =  f_{ij}^k \eta^i \eta^j \frac{\partial}{\partial \eta^k}  \right) .
\eeqn
Notice that we have not imposed any relations: this is a {\em commutative} dg algebra.

Consider the map of graded algebras
\beqn\label{eqn:quasi1}
\Phi \colon \CC[\eta^i] \to \CC\left[x_i,\xi_i, \frac{\partial}{\partial \xi_i}\right]
\eeqn
which sends $\eta^i \mapsto  \frac{\partial}{\partial \xi_i}$. 
We want to show two things: first that this map is compatible with the differentials in \eqref{eqn:big} and \eqref{eqn:ce1}, and second that it induces an isomorphism on cohomology.

To see that the map is compatible with the differentials, we observe that for each $j$ we have the commutator
\[
\left[x_i \frac{\partial}{\partial \xi_i} , \frac{\partial}{\partial \xi_j}\right] = f^j_{ik} \frac{\partial}{\partial \xi_i} \frac{\partial}{\partial \xi_k} .
\]
which is equivalent to the equation $\d \Phi(\eta^j) =\Phi (\d \eta^j)$. 
We point out that this is the key step that is different in the non-abelian case. 

To see that $\Phi$ is an isomorphism on cohomology we note the following. 
First, the algebra of closed elements is generated by $x_i$ and $\frac{\partial}{\partial \xi_i}$ subject to relations \eqref{eqn:relation} and \eqref{eqn:relation3}.
Next, just as in the abelian case we can argue that any element where $x_i$ appears is exact, hence cohomologically trivial. 
\footnote{This argument can be made rigorous by the use of a natural filtration, but we omit the technical details here.}  

Mathematically, the cochain complex \eqref{eqn:ce1} is called the {\em Chevalley--Eilenberg} cochain complex and is denoted $\clie^\bu(\fg)$. 
This complex computes so-called Lie algebra cohomology, which controls things like central extensions and anomalies in the context of QFT. 

Physically, we will think of the commutative dg algebra $\clie^\bu(\fg)$ as the BRST algebra encoding the local operators of a ghost field on flat space. 
In more familiar notation, the algebra is generated by the BRST ghosts 
\[
\eta^i = \fc^{i}
\]
of ghost number 1 endowed with a free anticommuting product.
The familiar BRST differential is precisely $\d = f_{ij}^k \fc^i \fc^j \frac{\partial}{\partial \fc^k}$.

Although we think about this example as arising from an underlying gauge theory, we work only at the level of the Lie algebra, and ignore subtleties related to the global form of the gauge group. 

In conclusion, we have seen that the Koszul dual of the associative algebra $U \fg$ is the commutative dg algebra of Chevalley--Eilenberg cochains $\clie^\bu(\fg)$
\[
(U \fg)^! \simeq \clie^\bu (\fg) 
\]
which is precisely the algebra of local operators of a ghost field valued in $\fg$. 

\subsection{A non-free example} 
\label{sec:LG1}

Before we move on to the explanation of why Koszul duality appears in physics, we compute the Koszul dual of one more algebra. The algebra in question naturally arises in the context of interacting field theories, namely as the algebra of local operators of the B-twist of a $\cN = (2,2)$ Landau-Ginzburg theory with target a vector space equipped with a superpotential.

Let $V$ be an $n$-dimensional vector space and consider the algebra 
\[
\CC[x_i, \psi^j] = \Sym^\bu (V^\vee) \otimes \wedge^\bu(V) 
\]
where $\{\psi^j\}$ is a linear basis for $V$ and $\{x_i\}$ is its dual basis.
The variable $x_i$ is even parity and $\psi^j$ is odd parity; we will furthermore take a $\ZZ$-grading where $x_i$ is degree zero and $\psi^j$ is degree $-1$. These operators have a physical origin as the $Q_B$-closed scalar and fermionic components of chiral multiplets valued in $V$.
Geometrically, this is the graded algebra of (polynomial) polyvector fields on $V$.

Consider a `superpotential' $W(x) = W(x_1,\ldots, x_n) \in \CC[x_1,\ldots, x_n]$ and the operator 
\[
(\partial_j W) \frac{\partial}{\partial \xi_j}
\]
acting on $\CC[x,\xi]$ by derivations. 
It is immediate to see that $\d_W \circ \d_W = 0$ and so we obtain the commutative dg algebra 
\[
\cA = \left(\CC[x_i, \xi^j] \, , \, \d = (\partial_j W) \frac{\partial}{\partial \xi_j} \right).
\]
Notice that in the case $W = x_1^2 + \cdots + x_n^2$ this complex is precisely the resolution we used in \eqref{eqn:res1}. 

The algebra of closed elements is generated by the variables $x_i$ in cohomological degree zero. 
Furthermore, any polynomial of the form $(\partial_i W) p(x)$ where $p$ is an arbitrary polynomial is exact for the differential. 
Thus, the cohomology of this complex is concentrated in degree zero, and is precisely the Jacobian ring of~$W$
\[
\CC[x_i] \, / \, (\partial_i W) .
\]

Of course, we have recovered the usual chiral ring of the B-twisted Landau-Ginzburg model. 

When $W = $ constant, the algebra $\cA$ has the following special property: it is self-Koszul dual. 
This means that $\cA^!$ is equivalent to the original algebra $\cA$:
\[
\cA^! = \left(\Sym^\bu (V^\vee) \otimes \wedge^\bu(V)\right)^! = \Sym^\bu (V^\vee)^! \otimes \wedge^\bu(V)^! \cong \wedge^\bu (V) \otimes \Sym^\bu (V^\vee) = \cA .
\]
Note that it is absolutely necessary to remember the full algebra $\cA$ for this statement to hold. 
For instance, the degree zero piece is simply the algebra of polynomials in the variables $x_i$, which is far from being self Koszul dual.

For more general $W$ the algebra $\cA$ no longer has this property. 
For instance, when $n=1$ and $W = x^3$, one has $\cA \simeq \CC[x]/(x^2)$. 
In this case the Koszul dual algebra is huge, given by the free (tensor) algebra on a single generator $x^!$: 
\[
\cA^! = {\rm Tens}(\CC) = \CC \oplus \CC x^! \oplus \CC (x^! \otimes x^!) \oplus \cdots .
\]

\section{Maurer--Cartan elements and algebraic couplings}
\label{s:algcouple}

In this subsection, we will explain how Koszul dual pairs of algebras naturally arise in the context of topological line defects: more precisely, when considering gauge-invariant couplings between (supersymetrically twisted) quantum field theories and topological quantum mechanical systems. We will largely expand upon the presentation in \cite{CP}; see also the appendix of \cite{GO}.

\subsection{Coupling two quantum mechanical systems} 

Before discussing the problem of coupling a line defect inside of a higher dimensional gauge theory we discuss the coupling of two one-dimensional theories whose algebras of local operators we denote by $\cA$ and~$\cB$. We will assume that the theories are topological, in the sense that their Hamiltonian/``time'' translation operators are cohomologically trivial with respect to their BRST operators $Q_\cA$ and $Q_\cB$; in this context, the algebras of local operators are (dg-)associative, since the operator product reduces to the associative multiplication. 

The completely `uncoupled' theory is the theory obtained by combining these two 
theories in a rather trivial way; at the level of local operators this is just the plain tensor product (over $\CC$):
\beqn\label{eqn:uncoupled}
\cA \otimes \cB .
\eeqn
The total BRST operator is simply $Q_\cA \otimes 1 + 1 \otimes Q_\cB$ \footnote{Although in this section we take a purely algebraic perspective, we remark that at the level of action functionals, which we will turn to later, this corresponds to simply taking the sum $S_{\cA} + S_{\cB}$.}.

We restrict to nontrivial couplings between theories $\cA$ and $\cB$ which are visible at the level of the BRST complex. 
This means that the coupled dg algebra of local operators is, as a graded algebra, exactly as in \eqref{eqn:uncoupled} (the plain tensor product over $\CC$). 
What is modified is the BRST operator: 
\beqn\label{eqn:coupled}
Q_{\cA} \otimes 1 + 1 \otimes Q_{\cB} + \til{Q}
\eeqn
where 
\[
\til{Q} \colon \cA \otimes \cB \to \cA \otimes \cB
\]
is some new operator contributing to the coupled BRST operator. 
In order for \eqref{eqn:coupled} to be of homogenous ghost number we see that $\til{Q}$ must also be of ghost number one.

We remark that the addition of such an operator $\til{Q}$ to the BRST differential is one way to formulate the notion of a coupling between theories with operator algebras $\cA$ and $\cB$. 
More generally, one can speak about deforming the tensor product $\cA \otimes \cB$ as a dg algebra (or even $A_\infty$ algebra), which is algebraically controlled by the Hochschild cohomology\footnote{The string theoretic reader may be familiar with taking Hochschild cohomology of a category of boundary conditions in passing from the open string to the closed string (as always, backgrounds of the closed string deform the open string algebra) \cite{DBranesmirror}. When the category of boundary conditions can be viewed as the category of modules over an associative algebra $A$, where $A$ is given by endomophisms of a generating object, we recover $HH(A, A)$. Koszul duality can be thought of as a transformation to another, dual generating object in such a category.}. 
Indeed, it is the second Hochschild cohomology $HH^2(A,A)$ of an algebra $A$ that controls infinitesimal (or first-order) $A_\infty$-deformations of the algebra $A$. 

We specialize the situation a bit further and make the following assumption: 
\begin{itemize}
\item the operator $\til{Q}$ is {\em inner} (or {\em Lagrangian}\footnote{It is precisely these elements which will correspond to Lagrangian density couplings via topological descent: see \S \ref{s:lagcoupling}.})
in the sense that it acts on an observable $\cO$ by
\[
\til{Q} \cO = [\alpha , \cO]_\star = \alpha \star \cO - \cO \star \alpha
\]
where $\alpha$ is an element in $\cA \otimes \cB$ and $\star$ is the product in the algebra $\cA \otimes \cB$. 
\end{itemize}
In order for this to make sense the element $\alpha$, like $\til{Q}$, must be of ghost number one. 

What other conditions must $\alpha$ satisfy so that the total BRST operator is well-defined? 
In the case of the purely uncoupled system, it is automatic that the BRST operator was square-zero $(Q_{\cA} \otimes 1 + 1 \otimes Q_{\cB})^2 = 0$ since both $Q_{\cA}$, $Q_{\cB}$ are and commute with one another.
For general $\alpha$ one has the following. 

\begin{lem}
The operator $Q_{\cA} + Q_{\cB} + [\alpha, -]_\star$ acting on $\cA \otimes \cB$ is square-zero if and only if $\alpha$ satisfies the following Maurer--Cartan equation 
\beqn\label{eqn:mc1}
Q_{\cA} (\alpha) + Q_{\cB}(\alpha) + \alpha \star \alpha = 0 .
\eeqn
\end{lem}

We refer to this as a Maurer--Cartan equation since it is equivalent to the condition that $\alpha$ is a Mauer--Cartan element in the dg Lie algebra whose underlying vector space is $\cA \otimes \cB$, differential is $Q_{\cA} + Q_{\cB}$, and Lie bracket is the commutator in the pure tensor product of algebras. 

From here on, we refer to $\alpha$ as an {\em algebraic coupling} (which we will soon relate to the notion of a {\em physical coupling}) between two quantum mechanical systems whose algebras of operators are given by ${\cA}, {\cB}$. 
For any such algebraic coupling we obtain a new dg algebra 
\[
\cA \, \underset{\alpha}{\til{\otimes}} \, \cB \define \left(\cA \otimes \cB \, , \, \d = Q_{\cA} + Q_{\cB} + \alpha \star (-) \right) 
\]
which is the BRST algebra of operators of the coupled system. 

We will now illustrate the notion of algebraic couplings in the context of a simple, yet fundamental, example. We will then connect this notion to both physical couplings and Koszul duality in the next sections. 

\subsection{Gauge coupling} \label{sec:gaugealg}

Fix a Lie algebra $\fg$ and
suppose that $\cA$ is the (commutative) algebra $\clie^\bu(\fg)$ of Chevalley--Eilenberg cochains of $\fg$. 
We think of $\cA$ as the dg-algebra of a BRST ghost in one-dimensional flat space. 

Next, we consider the Weyl algebra $\cW_V$ associated to a $\fg$-representation $V$. This is the algebra of operators of some topological quantum mechanics consisting of bosonic matter with a $\fg$-symmetry, and we intend to couple it to the theory given by the BRST ghost algebra. 

This will ultimately have the physical interpretation, per the discussion of \S \ref{s:lagcoupling},  of coupling the quantum mechanics to a 1d classical background (or flavor) gauge field. (In higher dimenions, if we had wanted to couple to a dynamical background gauge field and gauge the global symmetry of the matter theory, we would also include the usual $b$ antighosts in the BRST complex). 

We will see that the coupling to the ghost field is encoded by an explicit Maurer--Cartan element. 

Suppose $\{p^i\}_{i=1,\ldots,\dim V}$ is a basis for $V$ with dual basis $\{q_i\}$. 
The algebra $\cW_V$ is the free polynomial algebra on generators $p^i, x_i$ subject to the relation 
\[
[p^i, q_j] = \delta_j^i .
\]
In the notation of the previous section, we consider coupling the dg algebra $\cA$ to the algebra $\cB = \cW_V$. 
Notice that along with $V$, the algebra $\cW_V$ carries the natural structure of a $\fg$-representation where $\fg$ acts covariantly on the $p$-variables and contragrediently on the $q$-variables. 

Let $\{\fc_a\}_{a=1,\dots,\dim \fg}$ be a basis for $\fg$ and denote by $\{\fc^a\}$ the basis for $\fg^*$. 
In other words, $\fc^a$ is a linear BRST ghost for the gauge field. 
The $\fg$-representation $V$ is encoded by a linear map
\begin{align*}
\rho & \colon \fg \to {\rm End}(V) = V^* \otimes V \\
\rho(\fc_a) & = \rho_{aj}^i \, x_i \otimes p^j .
\end{align*}

Define the element $\alpha \in \cA \otimes \cB$ by the formula
\[
\alpha = \rho_{aj}^i \, \fc^a \otimes q_i p^j .
\]
We show that the element $\alpha$ satisfies the Maurer--Cartan equation \eqref{eqn:mc1}. 

First, notice that $Q_{\cA}$ is the Chevalley--Eilenberg differential 
\[
Q_{\cA} = f_{ab}^c \fc^a \fc^b \frac{\delta}{\delta \fc^c} .
\]
where $\{f_{ab}^c\}$ are the structure constants of the Lie algebra. 
Also, $Q_{\cB} = 0$ in this example. 

We compute
\[
Q_{\cA} \alpha = f_{ab}^c \rho_{cj}^i \fc^a \fc^b \otimes q_i p^j 
\]
and
\begin{align*}
\alpha \star \alpha & = \rho_{a,j}^i \rho_{b,\ell}^k \delta^j_k \fc^a \fc^b q_i p^\ell + \rho_{a,j}^i \rho_{b,\ell}^k \delta^\ell_i \fc^b \fc^a q_k p^j   \\
& = (\rho_{ak}^i \rho_{bj}^k  - \rho_{aj}^k \rho^i_{bk}) \fc^a \fc^b \otimes q_i p^l .
\end{align*}

Thus, we see that in order for the Maurer--Cartan equation to hold we must have $f_{ab}^c \rho_{cj}^i = \rho_{ak}^i \rho_{bj}^k  - \rho_{aj}^k \rho^i_{bk}$ for all $a,b,i,j$.
This is equivalent to the condition that $\rho$ be a well-defined $\fg$-representation. 

Before turning on the coupling $\alpha$, the decoupled system is of the form
\[
\cA \otimes \cB = \clie^\bu(\fg) \otimes \cW_V . 
\]
Turning on $\alpha$ has the effect of deforming this tensor product algebra to the algebra
\[
\cA \, \underset{\alpha}{\til{\otimes}} \, \cB = \clie^\bu(\fg ; \cW_V) ,
\]
where the right-hand side is the complex computing the Lie algebra cohomology of $\fg$ with coefficients in the $\fg$-representation $\cW_V$. 

Not all Maurer--Cartan elements are necessarily of this form. 
Since $\fg$ and $V$ are ordinary vector spaces, meaning in degree zero, the entire space of degree 1 elements in $\clie^\bu(\fg) \otimes \cW_V$ is $\fg^* \otimes \cW_V$. 
Such elements give rise to non-linear actions of $\fg$ on $\cW_V$. 

An analogous computation goes through when coupling to a fermionic quantum mechanical system, which furnishes the Clifford algebra associated to a $\fg$-representation.

\subsection{Other examples}

Let $p,q$ be two even variables satisfying the commutation relations
\[
[p,q] = 1.
\]
Denote by $\cW$ the resulting Weyl algebra generated by $p,q$.
 
Let $\psi$ and $\chi$ be two odd variables which satisfy the commutation relations
\[
\{\psi, \chi\} = \psi \chi + \chi \psi = 1.
\]
Denote by $\cC$ the (complex) Clifford algebra on these two odd generators.

Consider the algebra 
\[
\cA = \cW \otimes \cC 
\]
whose generators we label by $p_1,q_1,\psi_1,\chi_1$.
Consider also the (isomorphic) algebra
\[
\cB = \cW \otimes \cC
\]
whose generators we label by $p_2,q_2,\psi_2,\chi_2$.
We will produce an algebraic coupling between the algebras $\cA$ and $\cB$. 

Let $W(x_1,x_2)$ be a polynomial in two variables. 
Consider the odd element 
\[
\alpha_W = \chi_1 (\partial_{x_1} W)(q_1,q_2) + \chi_2 (\partial_{x_2} W)(q_1,q_2) .
\]
Since $\chi_i$ commutes with $q_j$ for all $i,j$, we see that $\alpha_W \star \alpha_W = 0$.
So, $\alpha_W$ is a Maurer--Cartan element in $\cA \otimes \cB$ for any $W$ and hence defines an algebraic coupling. 

Let's explicitly characterize the coupled system $\cA \, \underset{\alpha_W}{\til{\otimes}} \, \cB$ which as a cochain complex is 
\[
\big(\cW^{\otimes 2} \otimes \cC^{\otimes 2} , \d = \alpha_W \star (-) \big) .
\]
We compute the behavior of the differential on the generators of the algebra: 
\begin{align*}
\d q_i & = 0 , \quad i=1,2 \\
\d p_i & = \chi_1 (\partial_{x_1} \partial_{x_i} W)(q_1,q_2) + \chi_2 (\partial_{x_2} \partial_{x_i} W)(q_1,q_2) , \quad i = 1,2 \\
\d \chi_i & = 0, \quad i=1,2 \\
\d \psi_i & = (\partial_{x_i} W)(q_1,q_2), \quad i=1,2 . 
\end{align*}

In the case that the Hessian of the function $W$ is nonzero, we can relate this to an algebra that we introduced in Section \ref{sec:LG1}. 
Indeed, in this case it is easy to see that the only closed elements in our algebra are polynomials in $q_i$ and $\chi_i$. 
But, anything involving a single $\chi_i$ is exact by the second relation above. 
Thus, the cohomology is again the Jacobian ring of $W$:
\[
H^\bu \left(\cA \, \underset{\alpha_W}{\til{\otimes}} \, \cB\right) \cong \CC[q_1,q_2] / (\partial_i W)  .
\]
Observe that the algebraic coupling deformed us from a non-commutative associative algebra to one that is strictly commutative. 
\section{Physical couplings in QFT}\label{s:lagcoupling}

\subsection{Local operators}

Now that we have introduced algebraic couplings and sketched some elementary computations in homological algebra, we are ready to describe their quantum field theoretic context. First, we will introduce the notion of a physical coupling in the context of relevance for our discussion, and explain how it is related to the notion of algebraic coupling introduced in the previous section. Then, we will elucidate the appearance of Koszul duality. 

Let $\cB$ be the local operators of a one-dimensional translation invariant theory, and let $Q$ denote the BRST operator which acts on this space. 
The infinitesimal action by translations is through the time derivative $\frac{\partial}{\partial t}$. 

A one-dimensional theory is topological\footnote{In general dimensions, $Q$-exactness of all translation generators is the minimal condition for a cohomological theory to be called topological, and in some contexts this condition is referred to as weakly-topological. A stronger statement that many topological theories satisfy is $Q$-exactness of all components of the stress-energy tensor.} if there exists an endomorphism $\hat{Q}$ of the local operators, of cohomological degree $(-1)$, which satisfies
\beqn\label{eqn:triv}
\{Q, \hat{Q}\} = \frac{\partial}{\partial t} .
\eeqn
In words, translations along the line are trivial in BRST cohomology. 

The space of local operators $\cB$ is equipped with an algebra structure coming from the operator product expansion. 
Together with $Q$ we assume this equips $\cB$ with the structure of a  dg-algebra. In the case of an associative dg-algebra (dga), which has been our primary interest in this note, there is an associative multiplication and a differential which acts as a derivation on the associative product. 
More generally, $\cB$ may be endowed with higher operations which combine with $Q$ to form an $A_\infty$ algebra, and we will say a few remarks about this more general situation in what follows. 

\subsection{Topological descent}\label{sec:descent}
We are able to produce physical couplings from BRST complexes using the procedure of topological descent. The topological descent equations were first explored by Witten in the Donaldson twist of 4d $\cN=2$ theory \cite{W88} as a means to obtain higher-form-valued gauge-invariant observables from local operators of nonzero ghost number. The solutions to the descent equations were constructed more generally in \cite{MW97}; they may also be familiar to string theorists from their application to the BRST cohomology of noncritical string theories, e.g. \cite{WZ92}.

As above let $\cB$ be the local operators of a theory of topological quantum mechanics. 
We first recall how Equation \eqref{eqn:triv} solves the topological descent equations which allows us to construct a certain class of non-local operators out of local ones. Ultimately, through topological descent, we will obtain BRST-closed one-forms that, when integrated, can be added to the action of some topological quantum mechanics as a deformation \footnote{For topological theories in $d$-dimensions, one may similarly produce deformations to the original action by using topological descent to produce $d$-form-valued operators.}. The route to the BRST-closed condition will turn out to pass through the Maurer-Cartan equation discussed earlier. 

Because the theory is translation invariant, we can consider the complex of {\em differential form} valued operators
\[
\Omega^\bu(S \, , \, \cB) 
\]
where $S$ is any oriented one-manifold without boundary. 
This complex is equipped with the original BRST differential $Q$ acting on $\cB$ as well as the de Rham differential $\d$. 
This complex carries a total grading which combines the natural grading on differential forms and the internal grading / ghost number in the BRST complex $\cB$.  

Given a local (0-form) operator, $\cO = \cO^{(0)} \in \cB$, we can use the endomorphism $\hat{Q}$ to define the following one-form valued operator
\[
\cO^{(1)} \define (\hat{Q} \cdot \cO) \d t \in \Omega^1(\RR \, , \, \cB).
\]
The non-homogenous operator $\cO^{(0)} + \cO^{(1)}$ is closed in $\Omega^\bu(\RR \, , \, \cB)$ if and only if the following descent equations are satisfied 
\begin{align*}
Q \cO^{(0)} & = 0 \\
\d \cO^{(0)} + Q \cO^{(1)} & = 0 .
\end{align*}
If $\cO = \cO^{(0)}$ is of ghost number $k$ then $\cO^{(1)}$ is of ghost number $k-1$ but it also carries differential form degree one. 
So, with respect to the total grading on $\Omega^\bu(S, \cB)$ the operator $\cO^{(0)} + \cO^{(1)}$ does have a homogenous degree $k$.

We will use topological descent to produce Lagrangians that couple together two quantum mechanical systems. 
From the descent equations, we see that if $\{Q,\cO^{(0)}\} = 0$ then $\cO^{(1)}$ defines the following $Q$-closed Lagrangian
\[
\int_S \cO^{(1)} 
\]
where we integrate over any \textit{closed} \footnote{In general, we are interested in manifolds without boundary: if $S$ is a manifold with boundary, then there will be terms integrated over $\partial S$ that fail $Q$-closure.} oriented one-manifold $S$.
If $\cO^{(0)}$ is of ghost number $k$ then this Lagrangian is of ghost number $k-1$. 
Therefore, in order to pick out a standard (ghost number zero) Lagrangian one must start with a local operator of degree~$+1$. Thus, at the classical level, topological descent guarantees a BRST-invariant Lagrangian that one can add to the original Lagrangian as a deformation. 

At the quantum level, the Lagrangian gets upgraded to a path ordered exponential. Demanding that the coupled system satisfies BRST invariance to all orders in perturbation theory will force the coupling to satisfy the Maurer--Cartan equation.

The path-ordered exponential of a Lagrangian $\cL$ is of the form 
\[
\text{PExp} \int_\RR \cL .
\]
Starting with some local operator $\cO \in \cB$ of ghost number one, we apply this to the case $\cL = \cO^{(1)} = (\hat{Q} \cdot \cO) \d t$.
The BRST variation of the path integral is
\[
\int_{u=-\infty}^\infty \d u \left( \int_{-\infty}^{t = u} \d t (\hat{Q} \cdot \cO) (t) \right) \big(\partial_u \cO(u) + \Hat{Q} (Q \cO)(u) \big) \left( \int_{s=u}^{\infty} \d s (\hat{Q} \cdot \cO) (s) \right) .
\]
Integrating the $\partial_u(-)$ term by parts, we pick up two boundary terms corresponding to $u=s$ and $u=t$:
\[
\int_{u=-\infty}^\infty \d u \left( \int_{-\infty}^{t = u} \d t (\hat{Q} \cdot \cO) (t) \right) \big(\cO (u) (\Hat{Q}\cO) (u) - (\Hat{Q}\cO) (u) \cO (u) + \Hat{Q} (Q \cO)(u) \big) \left( \int_{s=u}^{\infty} \d s (\hat{Q} \cdot \cO) (s) \right) .
\]
We now recognize the Maurer--Cartan equation. 
Since $\Hat{Q}$ is a derivation, we observe that the last expression vanishes precisely when
\[
Q \cO + \cO \star \cO = 0 
\]
which is the Maurer--Cartan equation for $\cB$, thought of as a dg Lie algebra. 

Our analysis so far has been a little bit oversimplified. In particular, we assumed that the algebra of a topological quantum mechanics took the form of a dg algebra. 
When we work at the derived level, this is a very convenient assumption we can make without loss of generality. 
However, generically, when we are working with a theory that is topological in the sense of (\ref{eqn:triv}), it is useful to keep in mind that product between operators on the line may be nonsingular only up to $Q$-exact terms. 
Regularizing the divergences and keeping track of this more refined chain-level structure turns out to naturally produce an $A_\infty$-algebra---which is a structure which reflects the fact that the product may only be associative up to homotopy. 
The $A_\infty$ structure carries more information about configurations of operators in the original physical theory than passing to a quasi-isomorphic dg algebra. 
Carrying through an analogous analysis as above, by computing the condition for BRST invariance of the appropriately defined path-ordered exponential with a suitable regularization, then produces the Maurer-Cartan equation for an $A_{\infty}$ algebra (with different regularization choices related by quasi-isomorphisms \textit{of $A_{\infty}$-algebras})
\[
Q \cO + \mu_2(\cO, \cO) + \mu_3(\cO, \cO, \cO) + \ldots = 0,
\]
where $\mu_n(\cdot, \cdot, \ldots, \cdot)$ denote the $n$-ary operations of the $A_{\infty}$ algebra, with $\mu_1 = [Q, \cdot]$ and $\mu_2$ being the standard associative product.  

The deformed BRST differential is then modified by an infinite number of terms involving the higher products $Q \tilde{\cO} \mapsto Q \tilde{\cO} + \sum_{n, m}\mu_{n + m + 1}(\cO^n, \tilde{\cO}, \cO^m)$. The reader is encouraged to consult the appendices of \cite{GO} for more details, including explications of the $A_\infty$-algebra structure and the definition of the $n$-ary operations in terms of suitably regularized operator products on a topological line. 
In the remainder of the note, we will continue to implicitly assume that we are working with a strict dg algebra model and higher $A_\infty$ operations will play no role.

\subsection{From Maurer--Cartan elements to local couplings} 

We can now relate the notion of an {\em algebraic coupling} to the more standard notion of a local Lagrangian type coupling between quantum mechanical theories. 
The bridge is given by topological descent and is very similar in spirit to the last section where we only had a single algebra~$\cB$. 

Suppose that $\alpha$ is an algebraic coupling between two quantum mechanical systems $\cA$ and $\cB$. 
View $\alpha \in \cA \otimes \cB$ as a local operator inside of the uncoupled system. 
Since the uncoupled system is topological (as both $\cA$ and $\cB$ are) we can run topological descent to obtain the one-form operator $\alpha^{(1)} \in \Omega^1 (S , \cA \otimes \cB)$.
Further, we can form the Lagrangian 
\[
 \int_S \alpha^{(1)} .\]

By the previous discussion, we see that $\alpha$ is an algebraic coupling, meaning it satisfies the Maurer--Cartan equation for $\cA \otimes \cB$, if and only if $\int_S \alpha^{(1)}$ is BRST invariant (at least perturbatively) at the quantum level.

We turn now to an example. 

\subsection{One-dimensional gauge theory}
\label{sec:gaugetheory}

We return to the example of Section \ref{sec:gaugealg}, which clearly had some relationship to gauge theory, and finally flesh out the physics behind it. 

In this example $\cA = \clie^\bu(\fg)$ with $Q_{\cA}$ the Chevalley--Eilenberg differential of $\fg$.
As alluded to, this commutative dg algebra is equivalent to the dg algebra of local BRST operators of gauge theory on the real line. 
A gauge field is a one-form valued in the Lie algebra
\[
A \in \Omega^1(\RR ; \fg) .
\]
A ghost is simply a smooth $\fg$-valued function $\fc \colon \RR \to \fg$. 
The local BRST symmetry for the gauge field is
\[
\delta A^a = \d \fc^a + f_{bc}^a \fc^b A^c .
\]
Locally, any operator in the gauge field $A$ is exact for the linearized BRST differential. This is an application of the Poincar\'e lemma: all one-forms on $\RR$ are exact!
Thus, up to equivalence, all local operators can be built from the ghost $\fc$. 
The non-linear BRST operator acting on the ghost 
\[
\delta \fc^a = f_{bc}^a \fc^b \fc^c 
\]
is precisely the Chevalley--Eilenberg differential $Q_\cA$.

Since one-dimensional gauge fields only have topological degrees of freedom \footnote{As we remarked previously, in more general situations one must also include the Faddeev-Popov antighost, which in this setting is zero in cohomology.} we can run topological descent using the operator $\Hat{Q}_\cA$ defined by
\[
\Hat{Q}_\cA \fc^a = \iota_{\partial_t} A^a ,
\] where $ \iota_{\partial_t}$ is a contraction along the vector field generated by $\partial_t$. 
Notice that $\Hat{Q}$ commutes with the nonlinear part of the BRST operator. 
Further, this operator satisfies the commutation relation 
\[
\{\delta , \Hat{Q}_\cA\} = \frac{\partial}{\partial t} 
\]
and so satisfies the requisite data to run topological descent for the one-dimensional gauge field. 
Explicitly, applied to the ghost $\fc^a$, topological descent returns the one-form operator 
\[
(\fc^a)^{(1)} = \Hat{Q}_\cA \fc^a \d t = A^a  .
\]

The Weyl algebra $\cB = \cW_V$ (which has a trivial BRST differential $Q_{\cB} = 0$) is the algebra of operators of a system of topological mechanics on the real line. \footnote{In the case that one considers {\em complex} valued functions, one can think about this as the algebra of operators in a twist of $\cN=2$ supersymmetric quantum mechanics with a K\"ahler target $V$.}
The fields are pairs of functions $q  \colon \RR \to V$ and $p \colon \RR \to V^*$ and the action is 
\[
\int_\RR (p,\d q) = \int_\RR (p , \partial_t q) \d t 
\]
where $(-,-)$ is the natural symplectic pairing on $T^* V = V \oplus V^*$. 

This model is manifestly topological, with the equations of motion requiring $p,q$ to be constant functions.
This leads to the classical algebra of operators $\cO(T^* V) = \CC[q^i, p_j]$ where $i=1,\ldots, \dim (V)$. 
At the quantum level the Heisenberg commutation relations give the Weyl algebra $\cW_V$. 
 
Recall that the algebraic coupling defined by the data of a $\fg$-representation $\rho$ is $\alpha = \rho_{i,a}^j \fc^a q^i p_j$.
The one-form operator $\alpha^{(1)}$ is given by applying the descent procedure using the operator $\Hat{Q}$. 
Notice that in our setup here, $\Hat{Q}$ only acts on the gauge field.
We have
\begin{align*}
\alpha^{(1)} & = \Hat{Q} (\rho_{i,a}^j \fc^a q^i p_j) \d t  \\
& = \rho_{i,a}^j (\Hat{Q} \fc^a) q^i p_j \d t \\
& = \rho_{i,a}^j A^a q^i p_j ,
\end{align*}
leading to the Lagrangian $\int_{\RR} \rho_{i,a}^j A^a q^i p_j = \int_\RR p [A,q]$ which recovers the standard minimal coupling to a gauge field. 

The equations of motion for the fields $q,p$ are very simple: they are just required to be constant. 
Another approach, related to the Batalin--Vilkovisky (BV) formalism, is to impose the equations of motions of the fields $p,q$ cohomologically. 
That is, one can introduce one-form fields $q^{(1)}, p^{(1)}$ which satisfy $q^{(1)} = \d q$ and $p^{(1)} = \d p$. 
The equations of motion are simply $q^{(1)} = p^{(1)} = 0$. 
Doing this, we are treating the de Rham differential as a linear BRST operator $Q_{\cB} = \d$. 
Since the model is topological, there is a homotopy $\Hat{Q}'$ which satisfies $\Hat{Q}' q = \iota_{\partial_t} q^{(1)}$ and $\Hat{Q}' p = \iota_{\partial_t} p^{(1)}$, very similar to the case of the gauge field.

If we take into account the larger BRST complex, there are more terms present in the descendant operator $\alpha^{(1)}$:
\[
\alpha^{(1)} = \rho_{i,a}^j \left(A^a q^i p_j - \fc^a (q^i)^{(1)} p_j - \fc^a q^i (p_j)^{(1)}\right) . 
\] 
These additional terms involve the ghost field $\fc$, and will be present in any BV formulation of a 1d theory.  Indeed, while these couplings may look unfamiliar because of the explicit presence of the ghost field, the resulting couplings have the appropriate cohomological degree. The minimal BRST complex and its BV-BRST extension are quasi-isomorphic to one another. 

\section{Line defects}
\label{s:lines}
We extend the perspective on coupling of two one-dimensional systems to coupling a higher dimensional ``bulk'' theory to a theory of quantum mechanics along a one-dimensional line defect. 

The basic idea is very similar to the above discussion, though we will need to make some simplifying assumptions. 
We start with some bulk theory on a manifold of the form 
\[
\RR_t \times M 
\] 
where $M$ is some manifold. 
The crucial assumption we make is that the bulk theory is {\em topological} along $\RR \times \{x\}$ for any $x \in M$. 
This means that OPE's are independent of the coordinate $t$ \footnote{Again, this is generally true only up to $Q$-exact terms but our conclusions will not be affected by this simplification; we refer to the discussion of the previous section about the connection between dg-associative algebras and $A_{\infty}$ algebras.}.
An outcome of this assumption, is that the OPE of operators supported on $\RR \times \{x\}$, for any $x \in M$, have the structure of an associative algebra $\cA_x$. 
If we assume that the theory is also translation invariant along $M$ then there is also an equivalence (in our present context, a quasi-isomorphism) of algebras $\cA_x \simeq \cA_{x'}$ for $x, x' \in M$. 
This happens in the examples we discuss later on.
In this case, we drop the dependence on $x \in M$ and simply write the algebra as $\cA$. 

We will consider the problem of coupling to an auxiliary quantum mechanical system along $\RR_t$ whose algebra of local operators we will denote by $\cB$. 

Given these simplifying assumptions, the results of the previous sections characterize the coupling between $\cA_x$ and $\cB$.
We have the following.

\begin{prop}\label{prop:line}
There is a one-to-one correspondence between the space of couplings of $\cA_x$ to $\cB$ and the space of Maurer--Cartan elements in 
\[
\cA_x \otimes \cB .
\]
\end{prop}

We consider a simple example to illustrate this proposition more concretely. 

\subsection{Line defects in Chern--Simons theory}

As an example, let's consider perturbative Chern--Simons theory for a Lie algebra $\fg$ on $\RR^3 = \RR_t \times \RR^2$, and the problem of coupling a topological 
line defect along $\RR_t \times \{0\}$. 

Similar to the one-dimensional case, any local operator in Chern--Simons theory can be written, up to equivalence, as a function of a constant ghost field $\fc^a \in \fg^*$. 
The BRST operator acting on this field is simply $\delta \fc^a = f^a_{bc} \fc^b \fc^c$. 
Thus, in summary, the algebra of local operators for Chern--Simons theory is the Chevalley--Eilenberg complex 
\[
\cA = \clie^\bu(\fg) .
\]

If we treat Chern--Simons theory classically, then, the problem of coupling a one-dimensional line defect to Chern--Simons theory is identical to the problem of coupling a one-dimensional theory to a one-dimensional gauge theory for the Lie algebra $\fg$.
For topological quantum mechanics, see \S \ref{sec:gaugetheory}, we have described this coupling explicitly. 

What happens if we don't treat the bulk gauge theory classically? In perturbation theory, there is a single parameter for quantizing Chern--Simons theory on $\RR^3$ for a semisimple Lie algebra $\fg$: the Chern-Simons level. For an approach to this problem in modern mathematical language, see \cite[Chapter 5]{CosRenorm}. 
The discrete level can be modified by quantum corrections; restoring $\hbar$ for illustrative purposes, this in general leads to a perturbative series $k_0 + \hbar k_1 + \ldots$. In other words, at each order in perturbation theory one has a choice of level, which is a choice of an invariant symmetric bilinear form on $\fg$. Once one chooses a level for each order in $\hbar$, one can show that there is a unique $\hbar$-correction to the coupling $\alpha$ describing the line defect. 

\subsection{The universal line defect}
\label{sec:univ}

Proposition \ref{prop:line} characterizes the problem of coupling a fixed line defect to a higher dimensional theory. 
We now explain how the Koszul dual to the algebra of operators of the bulk theory represents the universal line defect. 

The key to this interpretation is the following idea.
Let's consider the problem of coupling a line defect to a simple bulk gauge theory. 
Assume that the algebra of local operators of the gauge theory is $\clie^\bu(\fg)$.
If the algebra along the line defect $\cB$ is concentrated in degree zero (so has trivial internal BRST differential), then couplings must necessarily lie in the subspace
\[
\alpha \in \fg^* \otimes \cB \subset \clie^\bu(\fg) \otimes \cB [1] .
\]
Thus, $\alpha$ can be thought of as a linear map $\phi_\alpha \colon \fg \to \cB$. 

The condition that $\alpha$ be a Maurer--Cartan element is $\d_{CE} (\alpha) + \alpha \star \alpha = 0$.
This condition is equivalent to a very natural condition on the map $\phi_\alpha$, that it is a map of Lie algebras:
\[
\phi_\alpha ([x,y]_\fg) = \phi_\alpha(x) \star \phi_\alpha(y) - \phi_\alpha(y) \star \phi_\alpha(x) .
\]

From here, we can see where Koszul duality enters. 
The characterizing property of the enveloping algebra $U \fg$ is that it is the universal algebra for which $\fg$ sits inside as a Lie algebra. 
This implies that there exists a unique map $\til \phi_\alpha$ making the following diagram commutative
\[
\begin{tikzcd}
\fg \ar[dr] \ar[rr, "\phi_\alpha"] & & \cB \\
& U \fg \ar[ur, "\til \phi_\alpha"'] .
\end{tikzcd}
\]

This spells out a bijective correspondence between couplings of a gauge theory to an external algebra $\cB$ along a line defect with maps of algebras $U \fg \to \cB$. 
The more general statement is the following.  

\begin{prop}
\label{prop:kd}
Let $\cA,\cB$ be dg algebras and suppose that $\ep \colon \cA \to \CC$ is an augmentation of $\cA$ and let $\cA^!$ denote the corresponding Koszul dual algebra. 
There is a bijective correspondence between:
\begin{itemize}
\item Maurer--Cartan elements 
\[
\alpha \in \cA \otimes \cB 
\]
satisfying $(\ep \otimes \id_B) (\alpha) = 0 \in \cB$.\footnote{Equivalently, this means that $\alpha$ lives in $\mathfrak{m} \otimes \cB$ where $\mathfrak{m} = \ker \ep \subset \cA$ is the augmentation ideal.}
\item Algebra homomorphisms 
\[
\phi_\alpha \colon \cA^! \to \cB .
\]
\end{itemize}
\end{prop}

For a textbook account of this result we refer to \cite[Chapter 3]{LV}.

A number of remarks are in order. 
We recognize part of the first condition above from the characterization of coupling $\cA$ to a line defect whose algebra of operators is $\cB$.
What does this extra condition that $\ep \otimes \id_B$ annihilate the algebra coupling mean, physically? 
To answer this, we need to explain the role of the augmentation.

The choice of an augmentation $\ep \colon \cA \to \CC$ corresponds to a physical choice of vacuum. 
Recall that the original bulk theory is defined on a manifold of the form $\RR_t \times M$. 
Consider the compactification of the theory on $\RR_t \times M \rightarrow \RR_t$. 
In particular, if $M$ is a noncompact manifold, one must choose suitable boundary conditions at infinity. 
If $\Obs(M \cup \infty)$ denotes the algebra of observables upon making this choice of boundary conditions, then there is a canonical map of algebras $\cA \to \Obs(M \cup \infty)$.\footnote{Using the formalism of factorization algebras this has the elegant interpretation as the structure map corresponding to the manifold embedding $\RR \times M \hookrightarrow \RR \times (M \cup \infty)$.} 

A simplifying assumption to construct an augmentation arises by making a suitable choice of boundary conditions which render the compactified theory {\em trivial}.\footnote{It would be interesting to determine whether or not all augmentations can arise in this way.}
In other words, the algebra of operators of the compactified theory is the trivial one-dimensional algebra $\CC$ and so there is a canonical map $\cA \to \CC$.
This is the augmentation. 

In theories with a mass gap, choosing boundary conditions at $\infty$ requires making a choice of massive vacuum and prescribing that the fields approach their expectation values in that vacuum as they approach the boundary. If the theory is massless, it violates our assumption that the compactified theory is trivial, so we exclude these theories from consideration \footnote{These massless theories fit in to a more general `relative' version of Koszul duality, which is a special case of the centralizer of a map of algebras. It will be interesting to make this connection precise in future work.}. One way to think about a ``mass gap'' in these topological theories is the following. Before passing to $Q$-cohomology, we study the theory at the level of chain complexes with differential $Q$. There, we may still access the generator of translations along the line, even though it is $Q$-exact, and we may similarly endow our manifold $M$ with some metric and study the spectrum of its Laplacian. The choice of a metric on $\RR \times M$ is merely a gauge choice, and we can study the spectrum of the resulting Hamiltonian. Since we are interested in the twisted theory with a topological line, every massive state is $Q$-exact (and we can make them infinitely massive by changing our gauge), and therefore trivial in the twisted theory; we exclude from our present consideration any theory with zero-energy states that are nontrivial in cohomology. 

With this physical interpretation of the augmentation, we can make sense of the condition that the algebraic coupling satisfy $(\ep \otimes \id_B) \alpha = 0$.
This simply says that the coupling preserves the choice of a vacuum.
This means that the coupled system corresponding to $\alpha$ leaves the vacuum unchanged.

We can finally make sense of what we mean by the ``universal'' line defect. 
Fix an augmentation, or choice of vacuum. 
Rather tautologically, there is a natural Maurer--Cartan element in the algebra 
\beqn\label{eqn:univ}
\alpha_{univ} \in \cA \otimes \cA^! 
\eeqn
which corresponds to the identity map $\id_{\cA^!}$ under the equivalence of Proposition \ref{prop:kd}. 
This coupling is universal, in the following sense. 
Any another coupling $\alpha \in \cA \otimes \cB$ can be obtained by
\[
\alpha = \phi_\alpha (\alpha_{univ})
\]
where $\phi_\alpha \colon \cA^! \to \cB$ is the corresponding map of algebras. 
We note that this universal property is sensitive to the choice of a vacuum; the element $\alpha_{univ}$ depends on this vacuum.  

\subsection{Universal line defect in Chern--Simons theory} 

A description of the universal line defect of perturbative Chern--Simons theory on $\RR^3$ follows from the discussions above. 

As we have already remarked, the algebra local operators of Chern--Simons theory on $\RR^3$ is equivalent to the Chevalley--Eilenberg complex $\cA = \clie^\bu(\fg)$. 

There is a natural choice of an augmentation at the classical level, namely the projection onto the unit observable $\ep\colon \cA \to \CC$. 
To see this how this augmentation persists at the quantum level, it is best to take the point of view of choosing a vacuum. 

One way to construct the augmentation is to place 3d Chern--Simons theory on the slab with a suitable choice of boundary conditions
\[
(x,y,t) \in \RR^2 \times [0,1] .
\]
We will see that with a clever choice of boundary conditions, the theory compactified along $\RR^2 \times [0,1] \to \RR^2$ (and hence upon further compactification to $\RR$) becomes trivial. 
This setup and boundary conditions were described in \cite{Aamand}, and we follow the presentation there. 

For simplicity, suppose that $\fg$ is an even dimensional Lie algebra. 
Then, choose subalgebras $\lie{l}_0, \lie{l}_1 \subset \fg$ such that 
\[
\dim \lie{l}_0 = \dim \lie{l}_1 = \frac{\dim \fg}{2} 
\]
and that the restriction of the trace pairing to both $\lie{l}_0,\lie{l}_1$ vanishes. 
In other words, $\lie{l}_0, \lie{l}_1$ are maximally isotropic subalgebras with respect to the trace pairing. 

The phase space for Chern--Simons theory at $\RR^2 \times \{0\}$ and $\RR^2 \times \{1\} \subset \RR^2 \times [0,1]$ is isomorphic to 
\[
\Omega^\bu(\RR^2) \times \fg[1] .
\]
If we demand that our fields have compact support, then the trace pairing endows this cochain complex with a symplectic structure. 
A boundary condition is specified by choosing a Lagrangian subspace of this symplectic space which is compatible with the non-linear gauge symmetries. 

At $\RR^2 \times \{0\}$ we choose the boundary condition to be the Lagrangian
\[
\Omega^\bu(\RR^2) \times \lie{l}_0 [1] \subset \Omega^\bu(\RR^2) \times \lie{g}_0 [1]  .
\]
In other words, we require that the gauge field satisfy $A|_{t=0} \in \Omega^1 (\RR^2) \times \lie{l}_0$ and similarly for the ghost and (in the BV formalism) anti-field. 
At $\RR^2 \times \{1\}$ we choose the boundary condition to be the Lagrangian
\[
\Omega^\bu(\RR^2) \times \lie{l}_1 [1] \subset \Omega^\bu(\RR^2) \times \lie{g}_0 [1]  .
\]
In other words, we require that the gauge field satisfy $A|_{t=1} \in \Omega^1 (\RR^2) \times \lie{l}_0$ and similarly for the ghost and anti-field. 

One can show that the compactification $\RR^2 \times [0,1] \to \RR^2$ results in a trivial theory on $\RR^2$, which is sufficient to guarantee a trivial theory on $\RR$ upon further compactification to 1d and satisfy our condition. 
There is a slight modification one can make in the case that $\fg$ is not even dimensional; we refer to \cite{Aamand} for more details. 
We point out that this vacuum gives a bit more structure than an augmentation as discussed in the previous section, where we only assumed that you obtain a trivial theory after compactification all the way down to one dimension. In fact, this slab setup constructs an augmentation for the local operators of Chern--Simons theory as an $\EE_2$-algebra, which  is used in \cite{Aamand} to relate the observables of Chern--Simons theory to the quantum group.

Finally, since the Koszul dual of $\clie^\bu(\fg)$ is 
\[
\clie(\fg)^! \simeq U \fg 
\]
we interpret the enveloping algebra $U \fg$ as the algebra of operators on the universal line defect for Chern--Simons theory on $\RR^3$. 

\section{Koszul duality: a free theory example}\label{s:free}

We turn to an illustrative example in three-dimensions of a free field theory on $\RR^3$. In spite of its simplicity, this theory is a useful toy model for illustrating the abstract considerations of the previous section. 
This theory we study is equivalent to the holomorphic twist of the $3d$ $\cN=2$ chiral multiplet. The details of this twist for general 3d $\cN=2$ theories appear in \cite{CDG, ACNV, ESW}, to which we refer the interested reader for details. 

The presentation of the holomorphic twist depends on the R-symmetry data we assign to the untwisted theory. 
The component fields of a 3d chiral multiplet of R-charge $r$ include a complex scalar $\phi$ of R-charge $r$ and four fermions $\psi_{\pm}$ with R-charges  $r-1$ and their conjugates $\bar{\psi}_{\pm}$ with R-charges $1-r$. We choose $r=0$ \footnote{Since the R-charge governs the twisted spin of the fields, more general R-charge assignments will result in the fields transforming as sections of certain powers of the canonical bundle: the $r/2$ power for the twisted scalar, and the $1-r/2$ power for the twisted fermion.}.  

The input data is simply a complex vector space $V$ in which the scalar field takes values. (In a more interesting bulk theory where the chiral multiplet were coupled to a $G$ gauge field, then $V$ would be a unitary linear $G$-representation of our choosing). Following \cite{DGP, CDG}, in preparation for the holomorphic twist, we can express the 3d $\cN=2$ superalgebra on $\CC \times \RR$ as $\left\lbrace Q_+, \bar{Q}_+ \right\rbrace = -2 i \partial_{\bar{z}}, \left\lbrace Q_-, \bar{Q}_- \right\rbrace = 2 i \partial_{z}, \left\lbrace Q_+, \bar{Q}_- \right\rbrace = \left\lbrace Q_-, \bar{Q}_+ \right\rbrace = i \partial_t.$

After the holomorphic twist with respect to the choice of supercharge 
\[
Q = \bar{Q}_+ 
\]
the surviving degrees of freedom are the scalar field and one of the fermions, $\bar{\psi}_-$. 
One can check from the supersymmetry transformations that these components are $Q$-closed, but not exact. 

The holomorphic twist involves a choice of coordinates of the form
\[
(t , z) \in \RR_t \times \CC_z \simeq \RR^3
\]
where $t$ is a real coordinate along $\RR$ and $z$ is a holomorphic coordinate along the remaining $\CC \cong \RR^2$ direction. 
The theory will be topological along $\RR$ but only holomorphic along $\CC_z$. 
This means that the OPE depends trivially on $t$ and holomorphically on $z$. 

One may perform descent on the two $Q$-closed fields of the twisted theory. If we restrict our attention to $\RR$ right away, following the discussion of the previous section, we can straightforwardly compute the descendants. For completeness, however, we note that there is a straightforward generalization of topological descent to the mixed holomorphic-topological case \cite{OY}.  We will employ this generalization here to compute three-dimensional holomorphic-topological descendants \cite{CDG}, which will restrict to our familiar 1d topological descendants on the line. 

If we denote the descendent 1-forms associated to an operator $\mathcal{O}$ by $\mathcal{O}^{(1)}$, then by construction these operators satisfy
\[
Q(\d z \wedge \mathcal{O}^{(1)}) = \d (\d z \wedge \mathcal{O})
\]
where $\d$ is the total de Rham differential on spacetime. Equivalently, we have $Q \mathcal{O}^{(1)} = (\d \bar{z}\partial_{\bar{z}} + \d t \partial_t)\mathcal{O}$. 
This equation can be solved for the descendent, and the solution is given by $\mathcal{O}^{(1)} = {i \over 2}Q_+ \mathcal{O} \d\bar{z} - i Q_- \mathcal{O} \d t$. 

Wedging the one-form descendent with $\d z$ produces a two-form operator that can be integrated along a two-cycle in spacetime to produce a $Q$-closed operator. 

Computing the corresponding descendants for the scalar and the fermion gives (up to signs and constants we suppress)
\begin{align}
\label{eqn:N=2desc1}
\phi^{(1)} &\sim \psi_{-} \d t + \psi_+ \d \bar{z} \\
\label{eqn:N=2desc2}
\bar{\psi}_-^{(1)} &\sim \partial_{\bar{z}}\phi \d t + \partial_t \phi \d \bar{z}.
\end{align}

The twisted theory is obtained by introducing the holomorphic supercharge $Q$ as a (linear) BRST operator. 
As shown in \cite{ACNV}, and reviewed in \cite{CDG,ESW}, one finds that in the BV formalism the twisted theory can be written as a first-order action.
Performing such a rewriting of the action will in general yield a different algebra of local operators than that of the original theory, but the two algebras will be quasi-isomorphic to one another (i.e. with identical $Q$-cohomology and any higher operations).

The scalar in the twisted, first-order theory is identical to the scalar in the physical theory with its twisted spin, and for $r=0$ they are simply identical. There is also a one-form $\psi$ in the twisted theory that may be identified, in the physical theory, with the descendant~$i \bar{\psi}_-^{(1)}$.

Hence, our fields consist of a scalar $\phi$ valued in $V$ and a partial one-form 
\[
\psi = \psi_t \d t + \psi_\zbar \d \zbar 
\]
valued in $V^*$, the dual vector space.
The free action reads
\[
\int_{\RR^3} \phi \partial_t \psi_\zbar + \int_{\RR^3} \phi \partial_\zbar \psi_t 
\]
which, after integration by parts, can be written more compactly in the twisted theory as\footnote{Although $\phi^{(1)}$ does not appear in the action, it plays an interesting role in understanding higher operations (brackets); see \cite{CDG} for details.}

\[
\int(\d z \wedge \psi) \wedge \d \phi.
\]

We may return to the physical theory for a free chiral multiplet from this twisted representation by adding a quadratic term for $\psi$ and integrating it out to return it to its on-shell value. 
Such quadratic terms are killed in the twisted formalism because they are $Q$-exact.  

The equations of motion for the scalar $\phi$ are
\beqn\label{eqn:phiEOM}
\frac{\partial}{\partial t} \phi = 0; \quad
\frac{\partial}{\partial \zbar} \phi = 0 .
\eeqn
The first equation implies that $\phi$ is constant in the $t$-direction and the second implies that $\phi$ is holomorphic in the variable $z$. 
The equation of motion for the one-form field $\psi$ reads
\beqn\label{eqn:psiEOM}
\frac{\partial}{\partial t} \psi_\zbar - \frac{\partial}{\partial \zbar} \psi_t = 0 .
\eeqn

There is an abelian gauge symmetry which acts on the $\psi$-fields, while leaving $\phi$ invariant.
The gauge symmetry is given by a $V^*$-valued scalar $\chi$ which acts on the fields by
\begin{equation}\label{eqn:freegauge}
\delta \psi_t = \frac{\partial}{\partial t} \chi; \quad
\delta \psi_\zbar = \frac{\partial}{\partial \zbar} \chi .
\end{equation}
This gauge symmetry in the twisted theory arises from the action of the holomorphic supercharge $Q$ on the fermion $\Bar{\psi}_-$. 

This ghost $\chi$ in the twisted theory is precisely $i \bar{\psi}_-$ in the physical theory. The BRST transformation of the twisted theory is identified with the SUSY transformation in the physical theory. Indeed, one may compute the SUSY transformation of $\psi$ and observe that it coincides with this symmetry on-shell.

It will be convenient to rewrite the formulas for the descendant fields in the twisted language. 
From Equations \eqref{eqn:N=2desc1}--\eqref{eqn:N=2desc2} for the untwisted descendants, we denote the twisted descendants as:
\begin{align}
\label{eqn:N=2desc1tw}
\phi^{(1)} &= \phi_t \d t + \phi_{\zbar} \d \bar{z} \\
\label{eqn:N=2desc2tw}
\psi^{(1)} &= \psi_{t \zbar} \d t \d \zbar  \\
\label{eqn:N=2desc3tw} 
\chi^{(1)} &= \psi .
\end{align}
In the BV formalism, the two-form $\psi_{t \zbar} \d t \d \zbar$ can be thought of as the anti-ghost to the ghost $\chi$. 
It will not play a role in any of the discussion below. 
Also, the one-forms $\phi_t \d t, \phi_{\zbar} \d \zbar$ combine to form an anti-field for $\psi$. 
We have also included the descendant for the ghost field $\chi$.
\subsection{Local operators}

We move towards a description of the algebra of local operators \textit{of all cohomological degrees} at $z = t = 0$, up to gauge equivalence. 
We recall that (local) operators are (local) functionals on the space of fields, so that while the field $\phi$ is valued in $V$, its local dual operator is valued in $V^*$. Although we are treating operators in this section, we will make use of the standard physics abuse of notation and denote them by the same symbols as the corresponding fields \footnote{Properly speaking, the R-charges we listed in the previous section were also those of the local operators. The fields have the opposite charges.}. 

The equations of motion imply that the $\phi$ has vanishing $t$ and $\zbar$-derivatives. 
Thus, the only local operators built from $\phi$ must be given by holomorphic $z$-derivatives. 
Let $V^*_{(n)} \simeq V^*$ denote the span of the local operators of the form $\partial_z^n \phi(0)$.
Thus, the linear local operators in $\phi$ comprise the vector space $\oplus_{n \geq 0} V_{(n)}^*$.
Since $\phi$ is a scalar, the algebra of operators involving only $\phi$ combine to form the symmetric algebra ${\rm S}^\bu \left(\oplus_{n \geq 0} V^*_{(n)}\right)$.

Next, consider local operators of the one-form field $\psi = \psi_t \d t + \psi_\zbar \d \zbar$. 
We claim that up to gauge equivalence there are no nontrivial local operators depending on $\psi_t, \psi_\zbar$.  
The easiest way to see this is to introduce the two-form field
\begin{align*}
\til \psi & \define \d z \wedge \psi \\ & = \psi_t \d z \wedge \d t + \psi_\zbar \d z \wedge \d \zbar .
\end{align*}
Then, the equations of motion are equivalent to the condition that this two-form be closed for the de Rham differential
\[
\d \left(\til \psi\right) = 0 .
\]
In other words, $\til \psi$ is required to be a closed two-form.

In terms of the two form field $\til \psi$ there is a one-form gauge symmetry by the one-form field 
\[
\til \chi \define \d z \wedge \chi .
\]
In this notation, the gauge symmetry simply reads $\delta (\til \psi) = \d (\til \chi)$, which is completely equivalent to Equation \eqref{eqn:freegauge}.
By a slight refinement of the Poincar\'{e} lemma for partial one-forms\footnote{This version of the Poincar\'{e} lemma is simply the statement $H^1_{\rm dR} (\RR) = H_{\dbar}^{0,1}(\CC)= 0$.}, there exists $\chi$ such that the one-form $\til \chi$ renders $\til \psi \in Z^2(\RR^3)$ gauge trivial.

Finally, consider local operators built from the ghost field $\chi$. 
The gauge transformations preserving $\psi = 0$ satisfy $\frac{\partial}{\partial \zbar} \chi = \frac{\partial}{\partial t} \chi = 0$. 
Thus, similarly to $\phi$, the local operators one can build from $\chi$ only involve holomorphic $z$-derivatives. 
Let $V_{(n+1)} \simeq V$ denote the span of the local operators of the form $\partial_z^n \chi(0)$.
Thus, the linear local operators in $\chi$ comprise the vector space $\oplus_{n > 0} V_{(n)}$.
Since $\chi$ is a gauge symmetry, the algebra of operators involving only $\chi$ combine to form the exterior algebra $\wedge^\bu \left(\oplus_{n > 0} V_{(n)}\right)$ \footnote{More precisely, we obtain the {\em graded symmetric} algebra $\Sym(V_{(n)}[-1])$ of $V_{(n)}$ concentrated in degree $+1$.
We will only use the underlying $\ZZ/2$-graded vector space in this section.}. Note that the operators coming from $\chi$ are concentrated in cohomological degree 1 and higher, as expected. Note that in the twisted theory, cohomological degree is identified with the R-charge of the twisted fields, which further coincides with ``ghost number'' in this example. 

Summarizing our discussion thus far, we see that the space of local operators at $z=t=0$ is
\beqn\label{eqn:freeops}
\cA_0 = \bigotimes_{n \geq 0} {\rm S}^\bu \big(V^*_{(n)}\big) \otimes \bigotimes_{m > 0} \wedge^\bu \big(V_{(m)}\big)
\eeqn
where $V_{(n)}^* \simeq V^*$ indicates the vector space spanned by local operators of the form $\partial^n_{z} \phi(0)$ and $V_{(m+1)} \simeq V$ indicates the vector space spanned by local operators of the form $\partial^{m}_z \chi$.

\subsection{Character of local operators}

We consider the following two symmetries of the free theory.
The first corresponds to conformal symmetries along the $z$-plane. 
Concretely, consider the piece of this generated by $\U(1)$ rotations of $\CC_z$.
Infinitesimally, this is generated by the holomorphic vector field $z \frac{\partial}{\partial z}$ acting on the fields $\phi$ and $\d z \wedge \psi$. 

Second, there is a $\U(1)$ flavor symmetry, where $\phi$ transforms with charge $+1$ and $\psi$ transforms in the dual representation, $-1$.

The character of the chiral algebra furnished by the space of local operators $\cA_0$ can be written in terms of fugacities $q$ and $s$ where $q$ is the fugacity of the $\U(1)$ conformal symmetry and $s$ is the fugacity of the $\U(1)$ flavor symmetry. The fugacity $q$ is conjugate to the twisted spin $J$, the fugacity $s$ is conjugate to the flavor charge $e$, and $r$, the R-charge, is the homological grading:
\[
\chi_0(q , s) \define \Tr (-1)^r q^J s^e =  \frac{(q s^{-1} ; q)_\infty}{(s ; q)_\infty}  = \prod_{n \geq 0} \frac{1 - s^{-1} q^{n+1}}{1 - s q^n}.
\]

This is nothing but the usual expression for the supersymmetric index of a free chiral multiplet. The degrees of freedom counted by this index are given by the $Q$-closed fields in the physical theory: $\left\lbrace \partial_z^n \phi \right\rbrace_{n \geq 0}, \left\lbrace \partial_z^n \bar{\psi}_- \right\rbrace_{n \geq 0}$. The former have $r=0, J = n$ and the latter have $r=1, J= n+1$. These perfectly match the local operators in the twisted theory generated by holomorphic derivatives of $\phi, \chi$ as enumerated in the previous subsection.

\subsection{Operators on a line defect}
As we have explained, we consider the coupling of our bulk theory to a topological line defect along $\mathbb{R}_t \times \{z=0\}$. 
Koszul duality will arise in terms of such line operators of the original three-dimensional theory. 
As this theory is free, the analysis here is quite simple, but it we hope it serves as an instructive exhibition of the general discussion in Section \ref{s:lines}. 

Let us consider the space of local operators $\cA_0$ at $z=t=0$. 
As announced in the general section, we are to interpret the Koszul dual of this algebra as the algebra of local operators on the ``universal line defect'' along
\[
\{z=0\} \subset \RR \times \CC .
\]
We proceed to compute this algebra from first principles.  

Recall that the theory is topological along this line. 
Along the line, the descendants of the fields $\phi, \psi$ are given by the $\d t$ components of the one-form descendants in the full theory \eqref{eqn:N=2desc1tw}--\eqref{eqn:N=2desc3tw}. 
The formulas for the descendant fields along $\{z=0\}$ are given as follows
\begin{align*}
\Hat{Q} \chi = \psi_t \\
\Hat{Q} \phi = \phi_t . 
\end{align*}
We are changing notations slightly to more easily line up with the fields of the twisted theory. 
Here $\psi_t$ is the $\d t$ component of $\psi^{(1)}$ and $\phi_t$ is the $\d t$ component of $\phi^{(1)}$. 
Notice that since $\chi$ is fermionic in the twisted theory, the field $\psi_t$ is bosonic (ghost number zero). 
Since $\phi$ is bosonic, the field $\phi_t$ is fermionic (ghost number $-1$).

Recall that if $\cB$ is the algebra of operators of an auxiliary quantum mechanical system along $\{z=0\}$, then couplings between $\cA_0$ and $\cB$ are controlled by Maurer--Cartan elements in $\mathcal{A}_0\otimes \mathcal{B}$.
The Lagrangian coupling is obtained by descent. 

Any Maurer--Cartan element lives in degree one, by definition. This leaves us with two options in our example. 
Let us introduce a basis $\{e_a\}$ of $V$ with dual basis $\{e^a\}$.
We may perform descent on the following two classes of local operators in $\cA_0 \otimes \cB$:
\begin{itemize}
\item $\partial_z^n \chi^a \otimes J_a[n]$ where $|\partial_z^n \chi^a| = 1, |J[n]|=0$, and $J = J_a e^a \in V^*$ and
\item $\partial_z^n \phi_a \otimes K^a[n]$, where $|\partial_z^n \phi_a|=0, |K[n]| = 1$, and $K = K^a e_a \in V$. 
\end{itemize}
At the moment, $J, K$ are stand-ins for some arbitrary local operators of the appropriate cohomological degree, and we will denote their descendants by a defect algebra endomorphism $\hat{Q}_{d}$ (which need not to be the same as the $\hat{Q}$ acting on the bulk algebra) by $J^{(1)}, K^{(1)}$. 
Notice that the first item is a presentation for the identity element $\id_{V_{(n)}} \in V_{(n)} \otimes V_{(n)}^*$ and so agrees with the form of the universal coupling \eqref{eqn:univ}. 
Similarly, the second item is a presentation for the identity element $\id_{V^*_{(n)}}$.

Performing descent, we see that the most general coupling takes the form
\[
\sum_{\ell \geq 0} \int_{\{z=0\}} \frac{1}{\ell!} K^{(1)} [\ell] \partial_z^\ell \phi + \sum_{k \geq 0} \int_{\{z=0\}} \frac{1}{k!} J[k] \partial_z^k \psi_t  + \sum_{j \geq 0} \int_{\{z=0\}} \frac{1}{j!} K[j] \partial_z^j \phi_t +  \sum_{i \geq 0} \int_{\{z=0\}} \frac{1}{i!} J^{(1)}[i] \partial_z^i \chi.
\] 

Notice that only the $\d t$-components of the one-forms $\psi, \phi^{(1)}$ may couple to a field along the defect, which we denote by $\psi_t, \phi_t$. This is equivalent to what we would have obtained had we performed ordinary topological descent directly on the line.
Also, observe that while the one-form field $\psi_t$ couples to a scalar $J[k]$, the field $\phi$ couples to a one-form along the line $K^{(1)}[\ell] = K_t [\ell] \d t$ (and vice versa for their counterparts). 
In total, all of the Lagrangian couplings are one-forms along $\RR_t$ with ghost number 0.

Only $\psi_t, \phi$ appeared in the original action of our twisted theory. 
On the other hand, the descendant field $\phi_t$ (which is a component of $\phi^{(1)}$) and the ghost $\chi$ appear in the coupling above. 
This is because local operators and their descendants are treated democratically when we generate the line defect couplings: we start with the tensor products of local operators (of any ghost number) from both the bulk and defect algebra and perform descent on the Maurer--Cartan elements, using the assumption that both the bulk and line defect theories are topological. 
If we start just with the physical fields, then the resulting set of couplings is slightly redundant.\footnote{The BV formalism provides an elegant way to avoid these redundancies from the outset, though we avoid using it in this note for clarity of presentation.} 
In other words, we must really work with the space of Maurer--Cartan elements of $\cA \otimes \cB$ \textit{modulo gauge redundancies}.

Translated to local operators along the line, this means that there are gauge redundancies and constraints among the space of local operators spanned by $J,K$ and their descendants.
Along the line $\{z=0\}$ one has the gauge symmetry $\delta \psi_t = \partial_t \chi$ and the descendant equation $\delta \phi^{(1)} = \partial_t \phi$. 
The gauge symmetry implies that there is a redundancy among the operators, namely $\delta K^{(1)}[j] = \partial_t K[j]$. 
Further, the descendant equation among the fields implies there is the following descendant equation among the operators $\delta J^{(1)} [k] = \partial_t J[k]$. 
In total, up to gauge equivalence we see that only the degree zero operators $\{J[k]\}$ and the degree $+1$ operators $\{K[j]\}$ remain.

Applying our analysis from \S \ref{s:lines} to obtain the Koszul dual algebra is almost trivial to compute in this case. 
The twisted BRST charge, $Q$, acts only on the bulk operators, which it annihilates. The operator products on the bulk terms, as we have seen, are just free products, so we have no interesting constraints on the OPEs for the defect operators from the quadratic terms: free products on the defect operators will satisfy the Maurer--Cartan equation. 
The operators $K[n], J[n]$ are anticommuting and commuting, respectively, per their cohomological degree above. In total we have

\begin{equation}\label{eq:freedual}
\mathcal{A}_0^{!} = \bigotimes_{n < 0} {\rm S}^\bu \big(V^*_{(n)}\big) \otimes \bigotimes_{m \leq 0} \wedge^\bu \big(V_{(m)}\big).
\end{equation} Notice that since the holomorphic derivatives increment the spin/degree of the operators, then in order for the coupling to have total spin 0, the gradings of the algebras are negative the gradings of the bulk algebras.

\subsection{Choosing a vacuum} 
In \eqref{eqn:freeops} we deduced the algebra of local operators $\cA_0$. 
Implicitly, in the above calculation of the Koszul dual of $\cA_0$ we assumed a natural augmentation given by the projection onto the ${\rm S}^0(V^*_{(0)})\simeq \CC$ piece of $\cA_0$. 
Indeed, each of the couplings listed above are annihilated by this augmentation. 
In \S \ref{sec:univ} we explained that such an augmentation arises physically from the choice of a vacuum. 
In this section, we explain this choice of a vacuum. 

This amounts to extending the theory on affine space $\RR \times \CC$ to a theory on $\RR \times \PP^1$, thinking of $\PP^1$ is the one-point compactification of $\CC$, by fixing the behavior for the fields at $z = \infty$. 
One subtle thing is that since $\PP^1$ does not admit a nowhere vanishing holomorphic volume form, it is most natural to locally replace $\psi$ by the two-form field $\til \psi = \d z \wedge \psi$.
Denote the $\d z  \wedge \d \zbar$ component of $\til \psi$ by $\til \psi_{z \zbar}$ and the $\d z  \wedge \d t$ component by $\til \psi_{zt}$. 

For the desired vacuum, we choose the boundary conditions on the \textit{fields} of the twisted theory:
\begin{itemize}
\item the scalar $\phi$ vanishes at $\infty \in \PP^1$. 
That is, $\phi$ is actually a section of the line bundle $\cO(-1)$ \footnote{More accurately this bundle should be written as $\pi^* \cO(-1)$ where $\pi \colon \RR \times \PP^1 \to \PP^1$ is the projection.}, 
\[
\phi \in \Gamma \left(\RR \times \PP^1 \, , \, \cO(-1) \otimes V \right) 
\]
\item the components of the one-form $\psi_{t}, \psi_{\zbar}$ have simple poles at $\infty \in \PP^1$.
That is, $\psi_t,\psi_{\zbar}$ are twisted by the line bundle $\cO(1)$,
\[
\til \psi_{zt}, \psi_{z\zbar}\in \Gamma\left(\RR \times \PP^1 \, , \, K_{\PP^1} \otimes \cO(1) \otimes V^*\right)
\]
\end{itemize}

Additionally, we must fix the boundary behavior of the ghost field $\chi$ to be compatible with the above choices.
Again, it is best to replace $\chi$ by the one-form field $\til \chi = \d z \wedge \chi$. 
We require that the one-form ghost $\til \chi$ acquire a simple pole at $\infty \in \PP^1$, just like the $\til \psi$-field:
\[
\til\chi \in C^\infty\left(\RR \times \PP^1 \, , \, K_{\PP^1}(1) \otimes V\right) .
\]

This choice of a boundary condition results in the following special property of the theory upon compactification along $\PP^1$.

\begin{lem}\label{lem:compactP1}
The compactification of the theory along 
\[
\begin{tikzcd}
\RR \times \PP^1 \ar[d] \\ \RR
\end{tikzcd}
\]
results in the trivial one-dimensional theory. 
\end{lem}
\begin{proof}
The point is that every solution to the equations of motion is gauge equivalent to the trivial solution. 
For the scalar $\phi$, notice that the equations of motion require it to be constant along $\RR$ and holomorphic $\dbar \phi = 0$. 
Since there are no global holomorphic sections of the sheaf $\cO(-1)$ on $\PP^1$, we see that $\phi$ must be the trivial solution. 

The remaining equations of motion are
\[
\frac{\partial}{\partial \zbar} \psi_t - \frac{\partial}{\partial t} \psi_\zbar = 0 ,
\]
which is equivalent to $\d \til \psi = 0$, where $\til \psi = \d z \wedge \psi$. 
We have already argued that before imposing boundary conditions that every pair of solutions to this equation are gauge trivial. 
Since the gauge transformations preserve the condition that $\psi_t, \psi_{\zbar}$ have a simple pole at $\infty$ this remains true even after imposing the boundary conditions. 

It remains to see that there are non trivial degrees of freedom present for the ghost $\til\chi = \d z \wedge \chi$. 
Recall that $\til\chi$ must be constant along $\RR$.
The remaining equations of motion together with our boundary condition implies $\til \chi$ is a holomorphic section of $K_{\PP^1}(1)$, but there are none. 
\end{proof}

If $\Obs(\PP^1)$ denotes the observables of the compactified theory, the above result implies that $\Obs(\PP^1) \simeq \CC$. 
The local-to-global map from local operators to operators in the compactified theory
\[
\cA_0 \to \Obs(\PP^1) \simeq \CC
\]
is precisely the projection onto the component ${\rm S}^0(V^*_{(0)})$. 

\subsection{Dimensional reduction}

Before moving on to a more interesting example, we take a small expository detour and demonstrate that one can obtain Koszul dual algebras from a rather different point of view. We will expand on this second point of view, and its connection to our main exposition, in \S \ref{s:conclusions}. On a first readthrough, the reader may wish to skip this analysis and jump directly to \S \ref{s:crit}. 

Consider, for a moment, the restriction of this free theory to the submanifold 
\[
\RR \times (\RR_{>0} \times S^1) \simeq \RR \times \CC^\times \subset \RR \times \CC .
\]
We will contemplate the dimensional reduction along $S^1$ factor and the resulting two-dimensional field theory on $\RR \times \RR_{>0}$. 
We will continue to denote the coordinate on the first $\RR$ by $t$, and will denote the coordinate on $\RR_{>0}$ by the circle radius $r$.  

For the naive dimensional reduction, we remember just the lowest lying $S^1$ modes. 
The resulting free theory simply consists of fields
\[
\phi_0(t,r) \in C^\infty(\RR \times \RR_{>0} , V) ; \quad \psi_0 (t,r) \in \Omega^1 (\RR \times \RR_{>0}, V^*) 
\] 
with action $\int_{\RR\times \RR_{>0}} \phi_0 \d \psi_0$.
There is, in addition, a topological gauge symmetry $\delta \psi_0 = \d \chi_0$. 

Of course, we have missed out on all of the higher $S^1$ modes present in the three-dimensional theory. 
More accurately what we should do is expand the fields $\phi, \psi$ into their Fourier modes around the circle.
We will then obtain a two-dimensional theory on $\RR \times \RR_{>0}$ with an infinite number of fields each labeled by their Fourier mode around the circle. 

Writing the complex coordinate as $z = r e^{i \theta}$ we obtain a decomposition of the three-dimensional fields on $\RR \times \CC^\times$:
\begin{align*}
\phi(t,z) & = \sum_{n \in \ZZ} \phi_n (t,r) e^{2 \pi i n \theta} \\
\psi(t,z) & = \frac{1}{2 \pi i z} \sum_{n \in \ZZ} \psi_n(t, r) e^{2 \pi i n \theta} .
\end{align*}

If we expand the free action, we see that the theory on $\RR \times \RR_{>0}$ has as its fields 
\[
\{\phi_n(t,r)\}_{n \in \ZZ} \in C^\infty(\RR \times \RR_{>0} , V) ; \quad \{\psi_n (t,r)\}_{n \in \ZZ} \in \Omega^1 (\RR \times \RR_{>0}, V^*) .
\] 
with free action
\[
\sum_{n \in \ZZ} \int_{\RR \times \RR_{>0}} \bigg(\phi_n \d \psi_{-n} + n \phi_n \psi_{-n} (r^{-1} \d r)  \bigg).
\]

Without the second term, this free action can be thought of as the topological $B$-model with target the algebraic loop space $L V = V[z,z^{-1}]$.
The second term includes a background gauge field for the symmetry of loop rotations.

There is an overall gauge symmetry of this dimensionally reduced action by an integers' worth of fields $\{\chi_n (t,r)\}_{z \in \ZZ}$ which acts via $\delta \psi_n (t,r) = \d \chi_n (t,r)$, where $\d = \d t \,\partial_t + \d r \,\partial_r$ is the two-dimensional de Rham operator. 

So far we have provided a description of the $S^1$ reduction of the theory on $\RR \times \CC^\times$ to $\RR \times \RR_{>0}$. 
The fact that this theory was defined on the affine space $\RR \times \CC$ determines a natural boundary condition for this dimensional reduction at $r = 0$, thus extending it to a theory on $\RR \times \RR_{\geq 0}$. 
This boundary condition reads
\begin{equation}\label{eqn:bc1}
\phi_n (t, 0) = 0 ; \quad \psi_{n+1} (t,0) = 0 \quad {\rm for} \; n < 0 .
\end{equation}
We must also require that the ghost field satisfy a compatible boundary condition:
\begin{equation}\label{eqn:gbc1}
\chi_{n+1} (t,0) = 0, \quad {\rm for} \; n < 0 .
\end{equation}

We observe that the corresponding operators at $t = r = 0$ of this dimensionally reduced theory agree precisely with the local operators $\Obs_0$ we found from the three-dimensional point of view above. 
Indeed, they are freely generated, up to gauge equivalence, by the even local operators $\{\phi_n\}_{n \geq 0}$ and the odd local operators $\{\chi_n\}_{n > 0}$. 
In the notation of the previous section $\phi_n$ corresponds to a scalar multiple of $\partial_z^n \phi$ and $\chi_n$ a scalar multiple of~$\partial^{n-1}_z \chi$. 

To summarize, imposing the boundary conditions (\ref{eqn:bc1}), (\ref{eqn:gbc1}) on the theory defined on the half-space $\mathbb{R} \times \mathbb{R}_{\geq 0}$ results in an algebra of boundary local operators given by (\ref{eqn:freeops}), the bulk line algebra in the original theory.  

\subsection{A ``transverse'' boundary condition}
We define another boundary condition of this dimensionally reduced two-dimensional theory which is naturally `dual' to the one we have just discussed. We will be a bit more precise about what we mean by this in \S \ref{s:conclusions} but, as we will see, the algebra of boundary local operators on this new boundary condition will reproduce the Koszul dual algebra (\ref{eq:freedual}).

In fact, we will actually construct this boundary condition at the level of the three-dimensional theory on $\RR_t \times \CC_z$ at $z = \infty$. 
This amounts to extending the theory on affine space $\RR \times \CC$ to a theory on $\RR \times \PP^1$ by fixing the behavior for the fields at $z = \infty$. 

The boundary conditions we choose are the following:
\begin{itemize}
\item the scalar $\phi$ vanishes at $\infty \in \PP^1$. 
That is $\phi$ is actually a section of the line bundle $\cO(-1)$ \footnote{More accurately this should be written as $\pi^* \cO(-1)$ where $\pi \colon \RR \times \PP^1 \to \PP^1$ is the projection.}, 
\[
\phi \in \Gamma \left(\RR \times \PP^1 \, , \, \cO(-1) \otimes V \right) 
\]
\item the components of the one-form $\psi_{t}, \psi_{\zbar}$ have simple poles at $\infty \in \PP^1$.
That is, $\psi_t,\psi_{\zbar}$ are twisted by the line bundle $\cO(1)$,
\[
\psi_{t}, \psi_{\zbar}\in \Gamma\left(\RR \times \PP^1 \, , \, \cO(1) \otimes V^*\right)
\]
\end{itemize}

Additionally, we must fix the boundary behavior of the gauge field $\chi$ to be compatible with the above choices.
A natural way to do this is to require the gauge field to acquire a simple pole at $\infty \in \PP^1$, just like the $\psi$-fields:
\[
\chi \in C^\infty\left(\RR \times \PP^1 \, , \, \cO(1) \otimes V\right) .
\]

This choice of a boundary condition results in the trivial 1d theory upon compactification along $\PP^1$, by the argument given above.

Let us translate this boundary condition to the $S^1$-reduction of the three-dimensional theory to $\RR \times \RR_{>0}$. 
This boundary condition determines the behavior of the two-dimensional fields as $r \to \infty$, which we read off to be 
\begin{equation}\label{eqn:bc2}
\phi_n (t, \infty) = 0 ; \quad \psi_{n+1} (t, \infty) = 0 \quad {\rm for} \; n \geq 0 .
\end{equation}
We must also require that the gauge field satisfy a compatible boundary condition:
\begin{equation}\label{eqn:gbc2}
\chi_{n+1} (t,0) = 0, \quad {\rm for} \; n \geq 0 .
\end{equation}

We can compute the gauge equivalence classes of local operators $\Obs_\infty$ at $r = \infty$ in a very similar way as we did for $r = 0$. 
The algebra is freely generated by the even local operators $\{\phi_n\}_{n < 0}$ and the odd local operators $\{\chi_n\}_{n \leq 0}$. 
In other words
\[
\Obs_\infty = \bigotimes_{n < 0} {\rm S}^\bu \big(V^*_{(n)}\big) \otimes \bigotimes_{n \leq 0} \wedge^\bu \big(V_{(m)}\big)
\]
which we can clearly identify as the Koszul dual algebra to $\Obs_0$, as anticipated.

It is also instructive to compare the characters of local operators at $r = 0, \infty$. 
We have already computed the character $\chi_0(q,s)$ of the algebra $\Obs_0$ above. 
At $r=\infty$ the character reads 
\[
\chi_\infty (q,s) = \prod_{n \geq 0} \frac{1 - s^{-1}q^{-n}}{1 - s q^{-n-1}}
\]
Notice that $\chi_{\infty}(q, s)$ converges in the opposite regime of $\chi_0(q, s)$ (the latter converging in the unit disc $|q|<1$), which is natural since their parent boundary conditions were defined respectively at $z=0, \infty$. To relate the two expressions, we employ the q-Pochhammer inversion formula, which reads $(z; p)_{\infty} = {1 \over (q z; q)_{\infty}}, \  p:= q^{-1}$. Application of this identity then gives that they are reciprocal as q-series: $\chi_0(q, s) = (\chi_{\infty}(q, s))^{-1}$. The fact that these characters multiply to the identity, representing the vacuum, is no accident. The two boundary conditions are \textit{transverse}: that is, they set to zero complementary ``halves'' of the degrees of freedom in the theory, though both boundary conditions preserve the vacuum corresponding to the choice of augmentation, $\CC_{\epsilon}$. 
This is in fact a general feature of Koszul dual algebras, and we will expand on the point of view that Koszul dual algebras are supported on transverse boundary conditions in \S \ref{s:conclusions}. It would be instructive to compute Koszul dual pairs of characters in more general theories, for which (by standard cohomological arguments applied to supersymmetric indices) it is generally sufficient to study free theories \footnote{We thank K. Costello for a discussion of this point.}.

We conclude by noting that it is already instructive to think about a degenerate limit of $S^1$ reduced theory. 
If we just remember the {\em lowest} lying $S^1$-modes, we have see that observables at $r=0$ are simply the free algebra on the even operator $\phi_0$, thus the subalgebra of lowest lying modes of the algebra $\Obs_0$ is
\[
\Obs_0^{(0)} = {\rm S}^\bu(V^*).
\]

On the other hand, at $r=\infty$ the lowest lying linear operator that survives is the {\em odd} operator $\chi_0$.
So, the subalgebra of lowest lying modes of the algebra $\Obs_\infty$ is
\[
\Obs_\infty^{(0)} = \wedge^\bu(V) .
\]

We recognize the prototypical example of Koszul duality for free algebras in $\Obs_0^{(0)}$ and $\Obs_\infty^{(0)}$.

\section{Koszul duality in 3-dimensional gauge theory}\label{s:crit}

Our second example will be another example of a theory that arises from a holomorphic-topological twist of 3d $\cN=2$ gauge theory. 
We treat only the pure gauge theory case; coupling to chiral multiplets (as discussed in the previous section) can be handled similarly. 
We will see that, unlike in the case of the free twisted chiral multiplet, the algebra of local operators will acquire nontrivial quantum corrections that are essential for Koszul duality to be implemented. 

When $\fg$ is semi-simple, one can interpret this holomorphic-topological twist as Chern-Simons theory at the ``critical'' level $k = h^{\vee}$, where $k > 0$ is the bare level, $h^{\vee}$ is the dual Coxeter number of the Lie algebra, and $k - h^{\vee}$ is the Chern-Simons level incorporating the one-loop shift from integrating out the vector multiplet gauginos. 

Physically, we begin with a vector multiplet with components $(A_{\mu}, \sigma, \lambda_{\pm}, \bar{\lambda}_{\pm})$ with canonical R-charge and spin quantum numbers, and will consider the components of the gauge field on $\RR_t \times \CC_{z}$ (though as explained in \cite{ACNV} the twisted construction to follow may be defined on any 3-manifold with a transverse holomorphic foliation). The gauge field takes values in some compact gauge group, though for the purposes of studying local operators we will only need to specify the Lie algebra $\fg$ in what follows. The vector multiplet also includes a real adjoint-valued scalar, $\sigma \in \fg_{\RR}$, and $\fg$-valued gauginos. 

With respect to the twisting supercharge $\bar{Q}_+$, it is easy to check that we have the following $\bar{Q}_+$-closed combinations:
\begin{align}
A & = A_{\zbar} (z,t) \d \zbar + A_{t} (z, t) \d t \\
\cA_t &= A_t - i \sigma.
\end{align}
The curvature of the connection 1-form in the first line is given by
\[
\cF_{\zbar t} := i  [(\partial_{\zbar} - i A_{\zbar}), \partial_t - i \cA_t] = F_{\zbar t} - i D_{\zbar}\sigma.
\] It is straightforward to check that $\overline{\cF_{\zbar t}}$ is also $\bar{Q}_+$-closed modulo the Dirac equation for the gauginos. 
The computation of the descendants of the local operators will proceed similarly to our previous example, except that now we use the covariantized descent equation:
\[
\bar{Q}_+ \cO^{(1)} = d_{A} \cO
\] with $d_A := (\partial_{\zbar} - i A_{\zbar})d\zbar + (\partial_t - i \cA_t)dt$. 

We will pass now to the twisted theory, where we write the theory in the first-order formalism and discard the $\bar{Q}_+$-exact terms. We will slightly relabel the relevant fields for Chern-Simons theory at the critical level. 

The theory has two sets of fields
\begin{align*}
A & = A_{\zbar} (z,t) \d \zbar + A_{t} (z, t) \d t \\
B & = B (z,t) \d z 
\end{align*}
with $A_{\zbar}, A_t$ $\fg$-valued functions on $\RR^3$ and $B(z,t) \in \fg^*$ is co-adjoint valued. If one rewrites the standard Yang-Mills action in the first-order formalism, one can identify $B_z \leftrightarrow {1 \over g^2} \overline{\cF_{\zbar t}}$ on-shell. The addition of a Chern-Simons term will modify this identification.

In terms of these fields, the twisted action can be written as a ``BF'' theory:
\begin{align*}
\int \; B \wedge F_A & = \int \; B \wedge \d A +  \frac12 \int \; B \wedge [A,A] \\
& = \int \left(B \partial_t A_{\zbar} + B \partial_{\zbar} A_t + B [A_t, A_\zbar]\right) .
\end{align*}
Here, we have implicitly used the canonical linear pairing between $\fg$ and its coadjoint module $\fg^*$. 
Notice that in the kinetic part of the action only $\partial_t$ and $\partial_{\zbar}$ derivatives appear since $B$ is a Dolbeault form of type $(1,0)$. At general Chern-Simons level, this action would be augmented by a Chern-Simons term using the partial connection: ${i \kappa \over 8 \pi} \int A \del A$ where $\kappa$ is the level.
Here $\del = \d z \partial_z$ is the holomorphic de Rham operator along $\CC$. 

\subsection{Gauge invariant local operators}

The Lie algebra of gauge symmetries is the familiar one: $C^\infty(\RR^3) \otimes \fg$, i.e. smooth functions valued in $\fg$. 
An element $\fc \in C^\infty(\RR^3) \otimes \fg$ acts on the fields via the transformation rule
\begin{align*}
\delta A^a_\zbar & = \partial_{\zbar} \fc^a + f^a_{bc} \fc^b A^c_\zbar \\
\delta A^a_t & = \partial_t \fc^a + f^a_{bc} \fc^b A^c_t  \\
\delta B_a & = f^c_{ab} \fc^b B_c .
\end{align*}

Forgetting the terms in the BRST operator coming from the Lie algebra structure on $\fg$, we can argue as in the last section that any local operator supported at $(z,t) = (0,0)$ which depends on $\partial^n_{\zbar} \fc^a$, $\partial^n_t \fc^a$, $\partial^n_{\zbar} B_a$, or $\partial^n_t B_a$ is BRST exact. 
Thus, only operators depending on $\partial_z^n \fc^a$ and $\partial^n_z B_a$ survive to BRST cohomology.

Since we are dealing with an interacting field theory, there is the possibility for quantum corrections.  
Because of this, we introduce the perturbative parameter $\hbar$. 
The quantum BRST algebra of local operators is defined over the ring $\CC[[\hbar]]$, which we will denote by $\cA$ in this section. 

As a graded algebra, the classical BRST complex $\cA / \hbar$ is a free polynomial algebra on an infinite number of generators
\[
\CC[\fc^a, \partial_z \fc^a , \partial^2_z \fc^a, \ldots, B_a , \partial_z B_a , \partial_z^2 B_a, \ldots], \qquad {a=1, \ldots, \dim \fg} 
\]
where $\partial_z^n \fc^a$ is of ghost number $+1$ and $\partial_z B_a$ is of ghost number $0$. 
The classical BRST operator take a similar form to the usual Chevalley--Eilenberg differential for the Lie algebra $\fg$ with coefficients in ${\rm S}^\bu (\fg)$. 
Indeed, it is the $\CC[[z]]$-linear extension of this operator.
The formula is
\[
Q_0 = \sum_{n +m \geq 0} \left(f_{ab}^c \partial_z^m \fc^a \partial_z^n \fc^b \frac{\delta}{\delta \partial_z^{n+m}\fc^c} + f_{ac}^b \partial_{z}^m \fc^a \partial_z^n B_b \frac{\delta}{\delta \partial_z^{n+m} B_c}\right) .
\] 
Notice that when restricted to the modes $\partial_z^n \fc^a, \partial_z^n B_a$ with $n=0$, this is the ordinary BRST differential. 

As alluded to, the BRST complex of local operators is equivalent to the Chevalley--Eilenberg complex for a certain graded Lie algebra.
The graded Lie algebra is
\beqn\label{eqn:grLie}
\fg[[z]] \oplus \fg^*[[z]] [1] .
\eeqn
This means that in degree zero there is the Lie algebra of currents $\fg[[z]]$ and in cohomological degree $-1$ there is the $\fg[[z]]$-module $\fg^*[[z]]$. 
Here, the module structure is simply the contragadient one where we also remember to add powers of $z$.
In bases, this reads
\[
[e_a z^n , e^b z^m] = z^{n+m} f_{ac}^{b} e^c  .
\]
The concise formula for the classical BRST algebra is, then
\[
\cA / \hbar = \clie^\bu \bigg(\fg[[z]] \oplus \fg^*[[z]] [-1] \bigg) .
\] 

Because $\fg[[z]]$ acts on $\fg^*[[z]][-1]$, this can be written in another way as a complex computing Lie algebra cohomology with coefficients in a module:
\[
\cA / \hbar \cong \clie^\bu \bigg(\fg[[z]] \, ; \, \Sym \left(z^{-1} \fg[z^{-1}]\right) \bigg) .
\]
In this identification, we have used the isomorphism between the (continuous) dual of power series $\CC[[z]]$ with the polar half of Laurent polynomials $z^{-1} \CC[z^{-1}]$ defined by the residue pairing. 
The $\fg[[z]]$ action on $z^{-1} \fg[z^{-1}]$ is given by the formula $e_a z^n \cdot e_b z^{-m-1} = f_{ab}^c e_c z^{n-m-1}$ if $n-m \leq 0$ and zero otherwise. 

Quantum effects modify the classical BRST operator as in
\[
Q = Q_0 + \hbar Q_1 + \cdots .
\]
In what follows, we will work only to first order in perturbation theory. 
At the end of the section we argue that this is actually enough: there exists a quantization of this three-dimensional theory which is exact at one-loop so we can assume $Q_k = 0$ for $k \geq 2$. 

We remark on a slightly different presentation of the algebra of operators. 
If $\fg$ is semi-simple the trace pairing defines an isomorphism of $\fg$-representations $\fg \simeq \fg^*$. 
In this case, we can write the the graded Lie algebra \eqref{eqn:grLie} in a more compact form as $\fg[\ep] [[z]]$ where $\ep$ is a formal parameter of cohomological degree $-1$. 
Then, in this case we have an isomorphism
\[
\cA / \hbar \simeq \clie^\bu \left(\fg[\ep][[z]]\right) .
\]

\subsection{Coupling to line operators}

By Proposition \ref{prop:line}, there is a one-to-one correspondence between line defects along
\[
\RR \times \{z=0\} \subset \RR^3 
\]
described by an associative algebra $\cB$ and Maurer--Cartan elements in 
\begin{equation}\label{eqn:criticalobs}
\cA / \hbar \otimes \cB = \clie^\bu \bigg(\fg[[z]] \, ; \, \Sym \left(z^{-1} \fg[z^{-1}]\right) \bigg) \, \otimes  \, \cB .
\end{equation}

For now, we are treating the bulk gauge field as classical.
So, we use the classical BRST operator $Q_0$ which is simply the relevant Chevalley--Eilenberg differential.
Thus, we are studying the Maurer--Cartan equation
\begin{equation}\label{eqn:cmccs}
Q_0 \alpha + \alpha \star \alpha = 0 .
\end{equation}

\subsection{An example: minimal free fermion coupling}

Before moving on to the computation of the universal line defect, we spell out the yoga between a class of Maurer--Cartan elements as in \eqref{eqn:cmccs} and local couplings in the case that $\cB$ is the algebra of local operators of the free fermion topological mechanics system. 

The fields of the quantum mechanics system consists of pairs of fermions $(\psi^i, \chi_i)$, $i=1,\ldots N$. 
In this case, the algebra of local operators $\cB$ is isomorphic to the Clifford algebra $\cC_{\RR^N} = \cC_N$ on generators $\psi^1,\ldots, \psi^N, \chi_1,\ldots \chi_N$. 
They satisfy the anticommutation relations 
\[
\psi^i \chi_j + \chi_j \psi^i = \hbar\, \delta^i_j .
\]

For now, let us treat the bulk gauge theory as classical. 
Then, couplings are given by Maurer--Cartan elements in 
\[
\cA / \hbar \otimes \cC_{N} .
\]
The most general degree $+1$ local operator in this algebra is of the $\alpha = \sum_n \alpha[n]$ where 
\[
\alpha [n] = \frac{1}{n!}\rho_{a,i}^{j} [n] \partial_z^n \fc^a \psi^i \chi_j 
\]
for some collection of coefficients $\{\rho_{a,i}^{j}[n]\}$. 
Here, $n$ can be any non-negative integer, $i,j=1,\ldots N$ and $a = 1,\ldots, \dim \fg$. 

Suppose that $\rho_{a,i}^j [n] = 0$ for $n \geq 1$. 
Since we are treating the bulk gauge theory as classical, the Maurer--Cartan equation \eqref{eqn:cmccs} reads
\beqn\label{eqn:mcz0}
f_{ab}^c \rho_{c,i}^{j} [0] \fc^a \fc^b \psi^i \chi_j + (\rho_{ak}^i [0] \rho_{bj}^k [0] - \rho_{aj}^k [0] \rho^i_{bk} [0]) = 0.
\eeqn
This computation is completely analogous to the one we did in \S \ref{sec:gaugealg} for the case of the Weyl algebra.
It is equivalent to the condition that $\alpha[0]$ define a map of Lie algebras 
\[
\alpha[0] \colon \fg \to \End(\RR^N) . 
\]
In other words, $\fg$ is represented on $\RR^N$. 

Let's go one step further. 
Assume that $\rho_{a,i}^j [n] = 0$ only for $n \geq 2$. 
In this case, in addition to \eqref{eqn:mcz1}, which involves only $\alpha[0]$, there are the additional terms in the Maurer--Cartan equation which requires
\begin{align}
\label{eqn:mcz1}
f_{ab}^c \rho_{c,i}^{j} [1] + (\rho_{ak}^i [0] \rho_{bj}^k [1] - \rho_{aj}^k [1] \rho^i_{bk} [0]) & = 0 \\
\rho_{ak}^i [1] \rho_{bj}^k [1]  & = 0 .
\label{eqn:mcz3}
\end{align}
Together, \eqref{eqn:mcz0}--\eqref{eqn:mcz3} imply that $\alpha[0] + \alpha[1]$ comprise a representation for the Lie algebra $\fg \otimes \CC[[z]]/z^2 \simeq \fg \ltimes \fg$. 

In general, one can extend the computations above to see that the full Maurer--Cartan equation requires that the coupling $\alpha = \sum_n \alpha[n]$ comprises a representation for the infinite-dimensional Lie algebra $\fg[[z]]$.

To construct the Lagragian coupling, we solve the descent equations $\partial_t \cO^{(0)} + Q \cO^{(1)} = 0$, where $\alpha = \cO^{(0)}$. 
Following the general discussion of topological descent in \S \ref{sec:descent}, we first construct an endomorphism $\hat{Q}$ acting on local operators which trivializes infinitesimal time translations: $\{Q, \hat{Q}\} = \partial_t$. 

In a way that's completely parallel to the construction of couplings in the case of a free theory, from $\alpha = \sum \alpha[n]$ we are led to the Lagrangian coupling
\[
\sum_{n \geq 0} \int_{\RR \times \{z=0\}} \frac{1}{n!} \rho_{a,i}^j[n] \partial_z^n A_t \chi_j \psi^i .
\]

\subsection{Tree-level Koszul duality and the universal line defect}

In this section we characterize the universal line defect of the 3-dimensional gauge theory.
We take the line to live along $\{z=0\} \subset \RR_t \times \CC_z$.

Similar to the case of the free theory in the previous section, the most general Lagrangian couplings take the form
\beqn\label{eqn:3dgaugeLag}
\sum_{n \geq 0} \int_{\RR \times \{z=0\}} \frac{1}{n!} J_a [n] \partial_z^n A^a_t + \sum_{n \geq 0} \int_{\RR \times \{z=0\}} \frac{1}{n!} K^a[n] \partial_z^n B_{t,a} .
\eeqn
Here $A_t$ is the $\d t$-component of the connection one-form and $B_t$ is the $\d t$-descendant of the field $B$. 
On the line, $B_{t}$ imposes the equation of motion $\partial_t B = 0$. 
We are ignoring the couplings of operators to the field $B$ and the ghost $\fc$ since they will not contribute to the cohomology of local operators on the line. 

The main goal of this section is to compute the universal line defect to tree-level in the perturbation parameter. 

\begin{prop}
At tree-level, the Koszul dual to the algebra of operators of the $3d$ gauge theory on $\RR \times \CC$ is the dg algebra generated by elements $J_a[k],K^a [\ell]$, $k,\ell \geq 0$, $a = 1,\ldots, \dim \fg$ of cohomological degree zero and degree $+1$, respectively, satisfying the relations
\begin{align*}
[J_a [k] , J_b [\ell]] & = f_{ab}^c J_{c} [k+\ell] \\
[J_a [k], K^b[\ell]] & = f_{ac}^b K^c [k+\ell] .
\end{align*} 
\end{prop}

From a mathematical standpoint this proposition is no surprise given our recollections on Koszul duality. 
Recall that the algebra of classical observables of the 3d gauge theory is the Lie algebra cohomology:
\[
\clie^\bu \bigg(\fg[[z]] \oplus \fg^*[[z]] [-1] \bigg)  
\]
Thus, the Koszul dual of the local operators is the associative algebra
\[
U\left(\fg[[z]] \oplus \fg^*[[z]] [-1] \right) .
\]
The description in the above proposition is a generators and relations presentation for this algebra.
Explicitly, the dictionary is 
\begin{align*}
J_a[n] & = z^n e_a \\
K_a [n] & = \ep z^n e^a 
\end{align*}
where $\{e_a\}$ is a basis for $\fg$ and $\{e^a\}$ is the dual basis for $\fg^*$.

Rather than relying on this general mathematical result about Koszul duality, we will deduce the relations of the Koszul dual algebra directly. 
We will do this by demanding that terms in the path integral involving the coupling to the line defect are BRST invariant, following the discussion in \ref{sec:descent}.

Let us first consider the classical commutation relation 
\begin{equation}\label{eqn:rel1}
[J_a [k] , J_b [\ell]] = f_{ab}^c J_{c} [k+\ell] 
\end{equation}
in the classical limit of the Koszul dual algebra. 

We will see that this relation arises from demanding that the piece of the path integral involving the first term in \eqref{eqn:3dgaugeLag} is BRST invariant.  
This term contributes the following expression to the path integral
\beqn\label{eqn:pexp1}
{\rm PExp} \left(\sum_{n \geq 1} \int_{\RR_t} \frac{1}{n!} J_a [n] \partial_z^n A^a_t \right) 
\eeqn
which can be expanded as a sum of terms of the form 
\beqn\label{eqn:pexp2}
\sum_{N \geq 0} \int_{t_1 \leq \cdots \leq t_N} \prod_{i=1}^N \left(\frac{1}{n_i!} J_{a_i}[n_i]  \partial_{z_i}^{n_i} A_t^{a_i} \right)  .
\eeqn

Recall that the BRST variation of the field $A_t$ is 
\[
\delta A^a_t = \partial_t \fc^a + f_{bc}^a \fc^b A_t^c .
\]
We readily compute the BRST variation of \eqref{eqn:pexp2} as
\beqn\label{eqn:brstvar1}
\sum_{N \geq 0} \sum_{j=1}^N \int_{t_1 \leq \cdots \leq t_N} \prod_{i=1, i \ne j}^N \left(\frac{1}{n_i!} \partial_{z_i}^{n_i} J_{a_i}[n_i] A_t^{a_i}   \right)  \left(\frac{1}{n_j!}  J_{a_i}[n_i] \partial_{z_j}^{n_j} \left( \partial_t \fc^{a_i} + f_{bc}^a \fc^b A_t^c \right)  \right)
\eeqn
Let's consider the two terms $\partial_t \fc^{a_i} + f_{bc}^a \fc^b A_t^c$. 
Integrating the term $\partial_t \fc^a$ by parts picks up boundary terms, when the points $t_j,t_{j+1}$ come together as we saw in the general analysis of the path integral in \S \ref{sec:descent}. 
When the two points collide, the corresponding operators are multiplied according to the algebra structure on local operators. 
As we saw, the two types of boundary terms contribute with alternate signs. 

Notice that we have not touched the term coming from $f_{bc}^a \fc^b A_t^c$. 
Thus, in order for this piece in the path integral to be BRST invariant we must have the following relation
\[
\int_t \sum_n \frac1{n!} J_a[n] \partial_z^n \left(f_{bc}^a \fc^b A^c_t \right) = \int_t \sum_{\ell,k} \frac1{k!}  \frac1{\ell!} [J_d[k] , J_e[\ell]] \partial_z^k \fc^d  \partial_z^\ell A^e   .
\] 
This relation must hold for all input fields since it is in the path integral.
To obtain the commutation relation \eqref{eqn:rel1}, we plug in the test fields
\begin{align*}
A_t & = e_b z^{\ell} \delta_{t=0} \\
\fc & = e_a z^k .
\end{align*}

Of course, there is a  Feynman diagrammatic interpretation of this computation. It is useful, particularly when going to higher orders in $\hbar$, to translate the above computation using the Feynman rules derived from  (\ref{eqn:3dgaugeLag}). The commutation relation is equivalent to the cancellation of the gauge variation of the tree-level Feynman diagrams in Figure~\ref{fig:cancel1}.

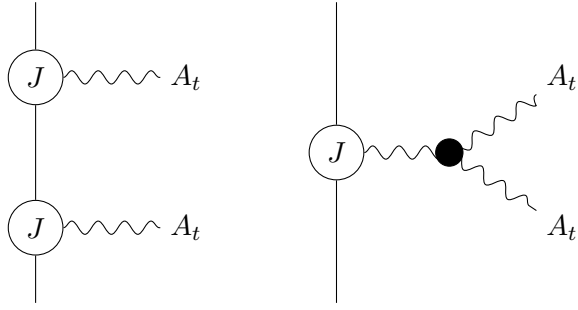
\begin{figure}
	\begin{tikzpicture}
	\begin{scope}
		\node[circle, draw] (J1) at (0,1) {$J$};
		\node[circle, draw] (J2) at (0,-1) {$J$};
		\node (A1) at (2,1)  {$A_{t}$};
		\node (A2) at (2,-1)  {$A_{t}$};
		\draw[decorate, decoration={snake}] (J1) --(A1);
		\draw[decorate, decoration={snake}] (J2) --(A2);
		\draw (0,2) -- (J1) --(J2) -- (0,-2); 
	\end{scope}	
	\begin{scope}[shift={(4,0)}];
		\node[circle, draw] (J) at (0,0) {$J$};
		\node (A1) at (3,1)  {$A_{t}$};
		\node (A2) at (3,-1)  {$A_{t}$};
		\node[circle,draw,fill=black, minimum size = 0.2pt]  (V) at (1.5,0) {};  
		\draw[decorate, decoration={snake}] (J) -- (1.5,0) --  (A1);
		\draw[decorate, decoration={snake}] (1.5,0) --(A2);
		\draw (0,2) -- (J)-- (0,-2);
	\end{scope}
	\end{tikzpicture}
	\caption{Cancellation of the gauge anomaly of these two diagrams leads to the equation for the commutation relations involving the currents $J[k]$.
	\label{fig:cancel1}}
\end{figure}

The tree-level commutation relation 
\begin{equation}\label{eqn:rel2}
[J_a [k] , K^b [\ell]] = f_{ac}^b K^{c} [k+\ell]
\end{equation}
can be derived in a completely analogous way so we omit the details. 

\subsection{Quantum corrections}

We now turn to quantum corrections to the universal line defect.
We have observed the Feynman diagrammatic interpretation of the classical, or tree-level, Koszul dual algebra. 
The quantum correction arises from anomalous diagrams present in the one-loop BRST variation of the couplings. 

The simplest one-loop Feynman diagram involving the couplings to the line defect is displayed in Figure \ref{fig:oneloop}. 
\begin{figure}
	\begin{tikzpicture}
	\begin{scope}[shift={(4,0)}];
		\node[circle, draw] (K) at (0,1) {$K$};
		\node[circle, draw] (J) at (0,-1) {$J$};
		\node (A) at (3,0)  {$A_{t}$};
		\node[circle,draw,fill=black, minimum size = 0.2pt]  (V) at (1.5,0) {};  
		\draw[decorate, decoration={snake}] (J) -- (1.5,0);
		\draw[decorate, decoration={snake}] (1.5,0) --(A);
		\draw[decorate, decoration={snake}] (K) -- (1.5,0);
		\draw (0,-2) -- (J) -- (K) -- (0,2);
	\end{scope}
	\end{tikzpicture}
	\caption{Simplest one-loop diagram involving coupling to the line defect.
	\label{fig:oneloop}}
\end{figure}
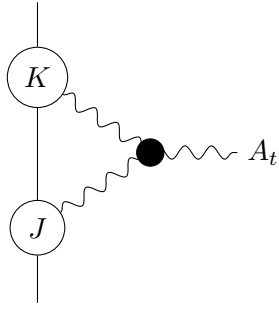

We will show that this one-loop diagram has a gauge anomaly of the form
\begin{multline}
\label{anomaly1}
\sum_{r,s} \frac{1}{(r+s+1)!} \int_{\RR \times \{z=0\}} \partial_z^{r+s+1} A_t^a f_{ac}^b J_{b} [r] K^c [s] \\ = \int_{\RR \times \{z=0\}} \partial_z A_t^a f_{ac}^b J_{b} [0] K^c [0] + \cdots .
\end{multline}
The ellipses denote similar terms present in the anomaly which involve a higher number of holomorphic $z$-derivatives.
Before moving onto the calculation of this anomaly, let's take a moment to reflect on its implications in terms of Koszul duality.

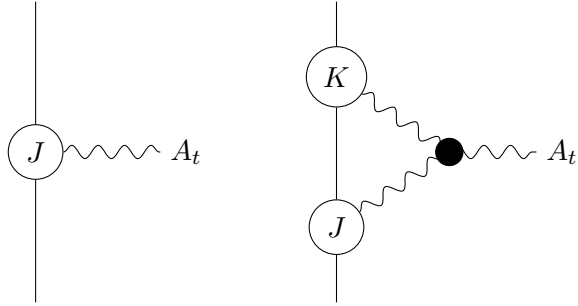
\begin{figure}
	\begin{tikzpicture}
	\begin{scope}
		\node[circle, draw] (J) at (0,0) {$J$};
		\node (A) at (2,0)  {$A_{t}$};
		\draw[decorate, decoration={snake}] (J) --(A);
		\draw (0,2) -- (J) -- (0,-2); 
	\end{scope}	
	\begin{scope}[shift={(4,0)}];
		\node[circle, draw] (K) at (0,1) {$K$};
		\node[circle, draw] (J) at (0,-1) {$J$};
		\node (A) at (3,0)  {$A_{t}$};
		\node[circle,draw,fill=black, minimum size = 0.2pt]  (V) at (1.5,0) {};  
		\draw[decorate, decoration={snake}] (J) -- (1.5,0);
		\draw[decorate, decoration={snake}] (1.5,0) --(A);
		\draw[decorate, decoration={snake}] (K) -- (1.5,0);
		\draw (0,-2) -- (J) -- (K) -- (0,2);
	\end{scope}
	\end{tikzpicture}
	\caption{Cancellation of the gauge anomaly of these two diagrams leads to the $\hbar$-linear correction to the differential $\d J[k]$.
	\label{fig:cancel3}}
\end{figure}

To cancel this gauge anomaly we must modify the algebra of operators on the universal line defect.
Classically, we have pointed out that the Koszul dual to the local operators of the bulk gauge theory is a graded associative algebra.
To cancel this gauge anomaly, we must introduce the following {\em differential} acting on this graded algebra (see Figure \ref{fig:cancel3} for the contributing diagrams):
\[
\d J_a[n] = \sum_{r + s = n-1} f_{ac}^b J_b[r] K^c[s] .
\]
If $k=0$ then $\d J_a [0] = 0$. 
The introduction of this differential gives us a consistent quantum theory. 
Indeed, if we introduce this differential into the gauge variation of the path-ordered exponential in \eqref{eqn:brstvar1} it precisely kills this one-loop gauge anomaly. 

There is a similar diagram where instead of a $J$ and a $K$ operator on the line, there are two $K$ operators, as in Figure \ref{fig:oneloop2}. 
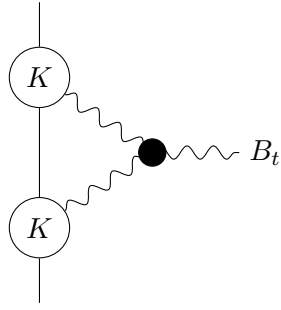
\begin{figure}
	\begin{tikzpicture}
	\begin{scope}[shift={(4,0)}];
		\node[circle, draw] (K1) at (0,1) {$K$};
		\node[circle, draw] (K2) at (0,-1) {$K$};
		\node (B) at (3,0)  {$B_{t}$};
		\node[circle,draw,fill=black, minimum size = 0.2pt]  (V) at (1.5,0) {};  
		\draw[decorate, decoration={snake}] (K2) -- (1.5,0);
		\draw[decorate, decoration={snake}] (1.5,0) --(B);
		\draw[decorate, decoration={snake}] (K1) -- (1.5,0);
		\draw (0,-2) -- (K2) -- (K1) -- (0,2);
	\end{scope}
	\end{tikzpicture}
	\caption{Another one-loop diagram with a gauge anomaly.
	\label{fig:oneloop2}}
\end{figure}
This diagram also has a gauge anomaly.
It is proportional to
\beqn
\int_{\RR \times \{z=0\}} \partial_z B_{t,a} f_{bc}^a K^b [0] K^c [0] + \cdots .
\eeqn

To cancel this gauge anomaly we must further modify the algebra of operators on the universal line defect.
We must introduce the following additional term in the differential which acts by the formula:
\beqn
\d K^a[n] = \sum_{r + s = n-1} f_{bc}^a K^b[r] K^c [s] .
\eeqn
If $k=0$ then $\d K^a [0] = 0$. 

\subsection{The computation of the one-loop anomaly}

We proceed with the calculation of the one-loop anomaly. 
First, we reflect on a particular symmetry present in the theory. 
Notice that the rescaling $z \to \lambda z$, for $\lambda \in \CC^\times$, the perturbative parameter scales as $\hbar \mapsto \lambda \hbar$. 
Thus, if a $k$-loop anomaly diagram takes as input operators involving $r$ holomorphic $\del_z$ derivatives then the anomaly will involve $(k+r)$ holomorphic $\del_z$ derivatives.

To compute the weight, we first introduce the propagator of the the holomorphic-topological theory.
For more details on the gauge condition used to construct this holomorphic-topological propagator we refer the reader to \cite{CostelloMtheory,GRWht}. 
Recall that the kinetic term in the action reads $\int_{\CC \times \RR} (B \del_t A_\zbar + B \del_\zbar A_t)$.
Thus, the propagator $P$ must satisfy the distributional equation
\[
\d z (\d t \del_t + \d \zbar \del_\zbar) P(z,t ; z',t') = \delta_{z=z', t=t'} \mathbb{1}_\fg  
\]
where $\mathbb{1}_\fg \in \fg \otimes \fg^*$ represents the identity map $\fg \to \fg$. 
An explicit solution to this equation on flat space can be constructed as follows. 
First, let 
\[
p(t,z) = \frac{\zbar \d t - t \d \zbar}{(|z|^2 + t^2)^{3/2}} .
\]
This is a one-form which is smooth away from the origin in $\CC \times \RR$. 
On $\CC \times \RR$ it can be thought of as a distributional one-form. 
The topological-holomorphic propagator on flat space is defined by
\[
P = j^* p \otimes \mathbb{1}_\fg 
\]
where $j \colon \CC \times \RR \to (\CC \times \RR)^2$ is the difference map $(z,t) \mapsto (z-z',t-t')$.
Since the propagator is symmetric across the diagonal, we will denote~$P(z,t;z',t') = P(z-z' , t-t').$

Now, let's consider the simplest case of the one-loop weight of the diagram in Figure~\ref{fig:oneloop}
\[
f^a_{bc} \int_{t_1,t_2, t, z} \d z \wedge P(t-t_1, z) \wedge P(t-t_2, z) \wedge A_{tz}^c (t,z). 
\]
This corresponds to taking as input the operators $J_a [0]$ and $K^b [0]$. 
From the paragraph above, we expect that the anomaly contains exactly one $\del_z$-derivative. 

Performing the $t_1,t_2$ integrals we can write the weight as 
\[
f^a_{bc} \int_{t, z} \d z [P'(z)]^2 \wedge A_{tz}^c (t,z)
\]
where $P'(z) = \int_{t} P(t,z) = \frac1z$ is the familiar holomorphic propagator on $\CC_z$.

The gauge variation of the weight is therefore 
\[
f^a_{bc} \int_{t, z} \d z [P'(z)]^2 \wedge \dbar A^c_t (t,z)
\]
To probe this anomaly, it suffices to input a test function of the form $A_t = z^n f(|z|) \delta_{t=0}$ where $f(|z|)$ is some radially symmetric bump function with $\int_{\CC} f(|z|) \d z \d \zbar =1$. 
By Cauchy's formula we see that the integral is nonzero if and only if $n = 1$.

More generally, we can replace $A_{tz}(t,z)$ by $\frac{1}{n!} \del_z^n A_{tz}(t,z)$ which corresponds to inputing the operators $J_a [r]$ and $K^b [s]$ with $r+s = n$. 
From this general case, we conclude that there is an anomaly for the diagram which is of the form 
\[
f_{bc}^a \sum_{r,s} \frac{1}{(r+s+1)!} \int_{t} J_a[r] K^b [s] \del_z^{r+s+1} A^c_{t} .
\]

The computation of the anomaly arising from the gauge variation of the diagram in Figure \ref{fig:oneloop2} involving two $K$-operators and the $B$-type field is completely analogous.

\subsection{Summary} 

In summary, we have found that in order for one-loop quantum gauge anomalies to vanish we must introduce a differential on the Koszul dual to the algebra of local operators.
In fact, we claim that this is an exact result, there are no higher loop anomalies!
This follows from the existence of a particular gauge which renders all two-loop and higher diagrams in the perturbative gauge theory identically zero \cite{GWcs}. 

We have shown that the quantum dg algebra takes the form
\[
\cA^! = \bigg( U\left(\fg[[z]] \oplus \fg^*[[z]] [-1] \right) \, , \, \d \bigg) .
\]
where on the $J$-generators the differential is
\[
\d J_a[n] = \sum_{r + s = n-1} f_{ac}^b J_b[r] K^c[s]  .
\]
On the $K$-generators it is
\[
\d K^a[n] = \sum_{r + s = n-1} f_{bc}^a K^b[r] K^c [s] .
\]

This quantum differential is related to a rich story in quantum groups, specifically the Yangian quantum group associated to the simple Lie algebra $\fg$. 
Costello argued in \cite{CYangian} that a certain four-dimensional partially topological, partially holomorphic gauge theory is ``controlled'' by the Yangian.
More explicitly, he shows that the Koszul dual to the algebra of local operators of the four-dimensional theory is equivalent to the (dual) Yangian as a topological Hopf algebra.  

If we place this four-dimensional theory on $\CC \times \RR \times S^1$ and compactify along $S^1$ then we precisely recover the theory on $\CC \times \RR$ studied in this section. 
A corollary of Costello's result is that the algebra of local operators of our theory is isomorphic to the dg algebra
\beqn\label{eqn:yangian1}
{\rm Hoch}_{\bu}\left( Y (\fg) \right)^\vee ,
\eeqn
where $Y(\fg)$ is the Yangian quantum group.
In words, this is the linear dual of the Hochschild {\em homology} of $Y(\fg)$. The differential is defined using the associative product on $Y(\fg)$, and the associative product on this dg algebra comes from the coproduct on $Y(\fg)$. 

Under Koszul duality, the quantum differential we find in this section becomes a quadratic relation in this Hochschild complex \eqref{eqn:yangian1}. 
This quadratic relation is rigorously understood in \cite{CYangian} as defining the Lie bialgebra structure on the classical limit of the (dual) Yangian. 

We expect, though do not prove here, that the quantum differential conspires, in fact, to produce a Koszul dual algebra equivalent to the original algebra: this gives a nontrivial example of a self-Koszul-duality. 
It is not obvious that this should be so (nor is it evident in the way we have written the algebras). 
Indeed, classically, $\cA/\hbar \neq \cA^{!}/\hbar$. 
It would be desirable to have a deeper explanation of this result.

\section{Conclusions and further directions}\label{s:conclusions}

In this note, we have explicated Koszul duality as an algebraic duality that arises when coupling a QFT to a topological quantum mechanical system. We would like to conclude by tying it to some parallel developments in the math and physics literature, as well as enumerate various open questions, works-in-progress, and future directions.

Of course, there are numerous technical questions that are raised by our presentation of Koszul duality that would be useful to clarify in future work. For example,
\begin{itemize}
\item As remarked upon in \S \ref{s:lines}, when one has a massless theory, some of our assumptions break down. For theories that do not admit a mass gap, it would be useful to understand the appropriate mathematical formulation of Koszul duality, likely related to relative, or \textit{curved}, Koszul duality \cite{curved}. 
For theories with a moduli space of vacua, we expect choices of augmentation to form continuous families, such that the resulting Koszul dual algebras should furnish interesting sheaves over moduli space. 

\item Similarly, but less ambitiously, it would be instructive to work through some examples with multiple massive vacua, and explore how the Koszul dual pairs of algebra vary given the explicit choice of augmentation/massive vacuum. 

\item In our examples, we have focused on Koszul dual pairs where one algebra was an ordinary degree 0 algebra and its dual was generated by elements in degree 1. 
This case is especially simple, in part because the extra augmentation constraint is essentially trivial, and it is easy to enumerate all Maurer-Cartan elements.  
It would be insightful to work out more explicit examples of algebras generated in multiple cohomological degrees, and deduce the corresponding space of Maurer--Cartan elements. 
Such examples can arise in theories in higher dimensions with higher form gauge symmetries, and so would be particularly interesting in the studies of higher dimensional TFTs and $\EE_n$-Koszul duality. 
It would also be interesting to understand if studying Koszul duality of these theories using their BV-BRST complexes of local operators can lend insight into the study of the higher group symmetries such theories enjoy \cite{2group, 2group6}.
\end{itemize}

Next, we will describe some versions of Koszul duality for a broader class of algebras than the (homotopy) associative algebras we have been focusing on. 
For some of these, the physical and mathematical generalization is relatively straightforward, while for others the appropriate definitions have not yet been made. 

\begin{itemize}

\item 
In the case of topological field theory, Koszul duality patterns have been observed in the context of Kontsevich formality in the works \cite{Shoikhet, Calaque, Rossi}. 
In \S \ref{s:lines} we alluded to the $A_\infty$-algebra (or $\EE_1$-algebra) structure present for algebras on topological lines (though see \cite{GO}). Also of interest are higher-dimensional topological defects (or, in codimension-1, boundary conditions) \footnote{Note that here we again mean a \textit{fixed} $d$-dimensional submanifold in spacetime along which translations are BRST-trivial. This is to be contrasted with the topological defect lines (TDLs) of, e.g., \cite{Yinetal}.}. For example, topological boundary conditions in 3d TQFTs have been studied in \cite{KS, BM, Kaidietal}. 

Topological QFT admits a functorial description in terms of $\EE_n$-algebras and their factorization homology \cite{AF1,AFT,LurieHA,Claudia}. 
Gauge-invariant couplings of QFTs to $n$-dimensional topological defects can be described using Koszul duality for $\EE_n$-algebras, and the latter is well-understood mathematically \cite{lurie2011formal}.
Further, in \cite{AFkoszul} Ayala and Francis have proved a duality for factorization homology of $\EE_n$-algebras  which simultaneously generalizes Koszul duality and Poincar\'e duality for manifolds. 

It would be interesting to see if Koszul duality affords an alternative characterization of these topological defects/boundaries as the space of $\EE_n$-homomorphisms from a Koszul-dual $\EE_n$-algebra. For instance, the perspective on Koszul duality as a ``Fourier transform" between transverse boundary conditions that we will discuss below may connect to the Lagrangian subgroup condition of \cite{KS}. 
It may also be fruitful to approach these problems using the BV-BFV formalism espoused in \cite{Cattaneo:2020lle} and elsewhere. 

\item Associative algebras and their $\EE_n$ cousins are natural algebras capturing OPEs of local operators in topological quantum field theories (including theories which are topological in the sense of cohomology).  Another concrete but rich algebraic structure capturing OPEs in certain (cohomological) field theories is that of a vertex algebra\footnote{For higher dimensional analogues of Kac-Moody algebras, relevant for capturing symmetries of holomorphically twisted QFTs, see \cite{GWkm}}. In contrast to the $\EE_n$ case, however, Koszul duality for vertex algebras has not been rigorously formulated \footnote{As in the associative algebra case, we are actually working with dg-vertex algebras}. Following the logic in this note, one would expect that a theory with holomorphic dependence on a plane (i.e. such that $\partial_{\bar{z}}$ is BRST-exact) has a Koszul dual vertex algebra corresponding to the universal holomorphic defect. One can try to compute this algebra by imposing BRST invariance of the resulting coupled system. This point of view was adopted to perturbatively compute Koszul dual vertex algebras in \cite{CP}. To formulate this notion mathematically, along the lines of this note, one needs the appropriate generalization of the space of Maurer-Cartan elements in the augmentation kernel. See \cite{Li} for some special cases of the requisite space of solutions.  

\item Throughout this note, we have focused on defects of order-type: that is, defects described by a lower-dimensional QFT with a local Lagrangian, coupled to the bulk QFT by local terms. It is natural to wonder if Koszul duality has some imprint on defects of disorder type, such as those describing nonperturbative objects in the bulk QFT. (Certainly, in systems with S-duality exchanging Wilson and 't Hooft lines, one can expect a connection). 

This question was answered in the special case of 't Hooft lines in 4d Chern-Simons theory in \cite{CGY}. There, the algebra of operators on an order-type line is the Yangian, which is Koszul dual to the bulk algebra \cite{CYangian}; it turns out that the Koszul dual of the operator algebra on a 't Hooft line is the dominant shifted Yangian (a subalgebra of the Yangian), predicated physically on consistency conditions arising when coupling the two types of defects together. 

It would be desirable to have a more general analysis of this situation. In particular, these operators source flux, which is expected to produce a deformation of the standard Koszul dual algebra via anomalous interactions with bulk modes. Rather than ask, then, that bulk/defect coupling is BRST-invariant, one must ask that the coupling has an anomaly that cancels the anomaly arising from the presence of flux escaping to infinity, so that the total system is gauge-invariant. A perturbative version of this analysis was employed in the vertex algebra case in \cite{CP}: there, the vertex algebras arose from the worldvolume theories of holomorphically twisted D-branes, which produced RR flux coupling to closed-string modes in Kodaira-Spencer theory via a local but non-BRST-invariant term. Mathematically, these deformations should connect to curved Koszul duality, which allows for a source term in the Maurer-Cartan equation \cite{curved}, and it would be useful to formulate this notion in generality. 

\item A general way to formulate local operator algebraic data in Euclidean QFT is via factorization algebras. See \cite{CG1, CG2} for an introduction. An analysis using the mathematical machinery of operads tells us that the Koszul dual of a factorization algebra is another factorization algebra, but computing these explicitly in non-twisted settings is a prohibitively difficult task. It would be instructive if a computable example could be found, even in a low-dimensional toy model. Many aspects of our formulation become murky. For example, we have been studying augmented algebras, which are symmetries that preserve a choice of vacuum; one may choose a nice symmetry-preserving asymptotic boundary condition in the UV, but have to contend with spontaneous symmetry breaking at low energies. Nonetheless, we will simply point out encouraging recent mathematical advances in understanding these foundational field theoretic phenomena; see \cite{EG} for a description of SSB in derived geometry. 

\end{itemize}

Finally, we wish to connect the discussion in this note to more general questions of recent physical interest, and highlight some alternative perspectives on Koszul duality that can be found in the literature. We hope that, in doing so, we have conveyed some of the beauty and utility of applying homological algebraic reasoning to various physical problems. Although homological algebra---as well as derived geometry, higher categories, etc.---may appear abstruse, it can serve as an organizational framework and computational toolbox for various physical questions.
\begin{itemize}
\item We explained in \S \ref{s:review} that Koszul dual algebras have the same derived category of dg-modules (roughly, representation categories), although we have not said much about the modules so far\footnote{For the cognoscenti, Koszul duality is a dg version of Morita theory.}. Let us make some physical and mathematical remarks now. One can draw a mathematical analogy between Koszul duality and Fourier transforms. In this analogy, $\cA, \cA^{!}$ correspond to canonically conjugate variables on a phase space like $x, p$, and the transformation that takes an object viewed as an $\cA$-module to its avatar as an $\cA^{!}$-module is analogous to a Fourier transform of tempered distributions. The role of the integral kernel in the Fourier transform is played by the Maurer-Cartan element.  More precise statements about module transformation in the dg-associative case can be found in \cite{GO}, along with a physical realization of this idea in the context of twisted, $\Omega$-deformed M2/M5-brane intersections. A module can be thought of in the topological line context as a defect on which lines can end. The algebra elements can act on this defect by translating them along the line to the endpoint where the local operator sits. A line may also puncture the defect, in which case the defect furnishes an algebra bimodule. One can obtain additional morphisms by studying defects as junctions between two different lines.  Additional structures, such as coproducts, on this category can also be obtained by fusing lines together, fusing defects together, and so on. In a set-up from a twisted version of M-theory, including intersections of line/surface defects arising from twisted M2/M5 branes, these questions were explored in detail \cite{GR}. Studying further examples, and providing a complete mathematical treatment of all these ingredients, is important future work. \\
\item In \S \ref{s:free}, we illustrated another way to obtain Koszul dual algebras from our free theory example: they arose as the algebras of local operators on two distinct boundary conditions upon compactification of our three-dimensional theory to a two dimensional half-space. In the resulting 2d TQFT, the original topological direction along which our putative line defect was supported lies transverse to the resulting boundary. This is, of course, no accident. Koszul dual associative algebras arise as local operator algebras on boundary conditions in 2d TQFTs \cite{DimofteStringMath}. More precisely, the two boundary conditions supporting the Koszul dual algebras are two distinct generating objects in the category of boundary conditions, and are sometimes called ``transverse''. The reason for this terminology is that, roughly speaking, transverse boundary conditions support complementary halves of the degrees of freedom in the 2d theory. This choice of coordinates/field variables for the degrees of freedom corresponds to a choice of phase space polarization spanned by the two ``transverse'' directions. This is natural from the point of view of Koszul duality as a Fourier transform we have sketched above. The boundary condition picture of Koszul duality was explored in 2d compactifications of 3d $\cN =4$ theories in \cite{BDGH}. Koszul duality in the Fukaya-Seidel derived category (i.e. the category of boundary conditions in the A-model) was also studied in \cite{Mina}, with applications to knot homology. 

Higher dimensional BPS boundary conditions may instead support, for instance, vertex algebras of local operators in a suitable twist. It would be very illuminating to explicate Koszul duality for vertex algebras from the boundary condition point of view as well. For example, examples of transverse pairs of BPS-boundary conditions in 3d $\cN=2$ theories were employed to construct duality interfaces in \cite{DGP}, and are also expected to be Koszul dual in the appropriate sense. Further, one should be able to connect a.) the transverse boundary condition picture of Koszul dual vertex algebras and b.) the picture employing gauge-invariant coupling of a bulk QFT to a holomorphic defect, again via a suitable dimensional reduction to a half-space. We leave to future work the task of making all of these connections more precise.\\

\item Costello and Li have used Koszul duality to provide proofs of powerful anomaly cancellation arguments in the topological string, even in situations where one must contend with an infinite number of possible counterterms or, more dramatically, nontrivial gauge anomalies at every order in perturbation theory. Homological algebra provides all-orders result proving that coupled open-closed string theories can nonetheless be quantized uniquely in the topological string. This analysis was performed in the B-model (BCOV theory) \cite{CLBCOV} and in a supersymmetric sector of the type I string \cite{CL2}. Analogous studies are underway for the open-closed A-model on hyperkahler spaces \cite{PS}. \\

\item Many of the ideas and ingredients we have touched on so far seem closely connected to quantization. Geometric quantization involves a choice of polarization on phase space, and one is often interested in the canonical transformations on phase space, while Koszul duality is a distinguished such symmetry (a swap of canonically conjugate variables given a choice of polarization). Furthermore, Koszul duality produces two alternative ways to view the category of boundary conditions in topological string theory, while studying A-branes on certain spaces provides an alternative geometric route to quantization \cite{GukovWitten}; see \cite{GaiottoWitten} for recent work connecting brane quantization and geometric quantization. Indeed, Koszul duality has been shown to be useful in the physical context of ``geometric quantization on the cylinder'', or geometric quantization of a relative bulk/boundary system. This setup has been fleshed out in detail \cite{EY} in the context of the TQFT relevant for Geometric Langlands program, and we will briefly summarize. Consider a TQFT theory on the ``cylinder'' $M \times \mathbb{R} \times \mathbb{R}_{\geq 0}$, with some choice of boundary condition (so, in particular, the direction normal to the boundary is topological). A line operator inserted in the topological direction and supported at a point on $\mathbb{R}_{\geq 0}$ can act on the boundary condition by bringing the point close to 0, wherein it fuses with the original boundary condition to form a new one. One can of course consider the trivial line, whose endomorphisms are simply local operators in the theory, which then act naturally on the category of boundary conditions. 

On the other hand, one can consider the TQFT on $M \times \mathbb{R} \times \mathbb{R}_{\geq 0}$ with a choice of boundary condition, and study the algebra of operators on $M \times U$ where $U$ is a contractible open set of the half-plane. Since we are in a TQFT we can use the state-operator correspondence to map this configuration to a codimension-1 strip with the same boundary condition at both ends; by the usual TQFT voodoo, this produces a Hilbert space of states with the structure of an algebra. This picture, illustrated in \cite{EY}, enables one to think of the endomorphism algebra of local operators on any given boundary condition of the TQFT as the quantum phase space of the bulk/boundary system. The algebra of local operators, viewed as endomorphisms of the trivial line, acts on any such boundary condition. This gives a map of algebras from the algebra of local operators to the endomorphism algebra of any boundary condition in the category. This must be true for any category object/boundary condition, so the algebra enjoys a universal property in the category of boundary conditions.

Both of these description of the category fit in naturally with the physical avatars of Koszul duality we discussed. This category can be thought of as a category of $\cA$-modules, where $\cA$ is the endomorphism algebra of the generating boundary condition. However, it has a (more tractable) Koszul dual description as a category of $\cA^{!}$-modules. The Koszul dual frame was used to nice effect in \cite{EY}.

An immediate and very interesting extension would be extending this setup to theories that are not TQFTs but rather are twisted theories with holomorphic or mixed holomorphic-topological dependence in the directions transverse to $\mathbb{R}$. The mathematical formalism relevant for these explorations is \textit{shifted geometric quantization}, reviewed in \cite{Safronov}. \\

\item At last, we come to the end of our tour through Koszul duality and its myriad guises. 
Naturally, we must conclude with the killer application of Koszul duality in physics: holography. More precisely, we claim that Koszul duality (suitably generalized), governs bulk/boundary algebras of operators in examples of the \textit{twisted holography} program, initiated in \cite{CG}. One may motivate this correspondence somewhat picturesquely using the coupled bulk/defect picture of Koszul duality. One can consider the open/closed string origins of the AdS/CFT correspondence and derive the constraints on this system from imposing BRST-invariance of their couplings, with a suitable deformation to the algebras arising from the D-brane backreaction. It would be interesting to translate the transverse boundary picture of Koszul duality into a twisted holography context as well; it would be fascinating if the transverse boundaries connected, for instance, to the idea of Gopakumar that every closed string theory has two open string duals \cite{Gopakumar}. The protoypical example of this is the FZZT/ZZ branes dual to 2d gravity coupled to Liouville theory/minimal models. This example begs for a twisted holographic interpretation. 

The idea that Koszul duality could be relevant for the holographic correspondence was initiated by Costello and Li in \cite{TwistedSugra}, and elaborated upon by Costello in \cite{CostelloM2} using his twisted M-theory construction \cite{CostelloMtheory}. Koszul duality in a twisted instance of the AdS/CFT correspondence, the duality between IIB string theory on $AdS_3 \times S^3 \times T^4$ and the $Sym^N(T^4)$ CFT, was explicated in \cite{CP}. More precisely, the authors proposed that deformed/curved Koszul duality for vertex algebras enabled one to obtain $N \rightarrow \infty$ OPEs of the boundary CFT from the BV-BRST complex of the closed string field theory on the appropriate spacetime in the bulk. More work needs to be done in order to fully understand and exploit this relationship. 

The utility of Koszul duality and twisted holography is not limited to AdS/CFT: celestial symmetry algebras that have been unearthed in the celestial holography program can be derived simply from Koszul duality considerations in the context of twisted holography on twistor space \cite{CP2}. 
\end{itemize}

The applications of these mathematical ideas to physics, and especially holography, are just beginning to bear fruit. We hope this note has served as a helpful introduction to this vast web of ideas.

\printbibliography


\end{document}